\documentclass[12pt]{article}

\usepackage[T1]{fontenc}

\usepackage{amsthm, amssymb, bm, bbm}

\usepackage{mathtools}
\usepackage{natbib}
\usepackage[pdftex]{graphicx}
\usepackage{appendix}
\usepackage{caption}
\usepackage{subcaption}
\usepackage{url}
\usepackage[noend]{algpseudocode}
\usepackage{booktabs}
\usepackage{float}
\usepackage[usenames,dvipsnames,svgnames,table]{xcolor}
\usepackage{multibib}
\usepackage[colorlinks=true,
            linkcolor=red,
            anchorcolor=blue,
            citecolor=MidnightBlue!70,
            urlcolor=black,
            hypertexnames=false]{hyperref}

\usepackage{setspace}
\usepackage{chngcntr}
\def\spacingset#1{\renewcommand{\baselinestretch}%
{#1}\small\normalsize} \spacingset{1}

\usepackage{tikz}
\usetikzlibrary{bayesnet}
\usepackage{tikzscale}

\usepackage[labelformat=simple]{subcaption}

\newtheorem{theorem}{Theorem}
\newtheorem{lemma}{Lemma}

\usepackage{mathtools}
\usepackage{mathrsfs}
\usepackage{float}
\usepackage{bm}

\allowdisplaybreaks
\def\bal#1\eal{\begin{align}#1\end{align}}
\def\balnn#1\ealnn{\begin{align*}#1\end{align*}}

\floatstyle{boxed}
\newfloat{sampler}{thp}{los}
\floatname{sampler}{Algorithm}

\DeclareMathOperator*{\argmin}{arg\,min}

\DeclareMathOperator*{\diag}{diag}
\DeclareMathOperator*{\tr}{Tr}
\DeclareMathOperator*{\pawd}{PAWD}
\DeclareMathOperator*{\vecop}{vec}

\DeclarePairedDelimiter\set{\{}{\}}
\DeclarePairedDelimiterX{\norm}[1]{\lVert}{\rVert}{#1}
\DeclarePairedDelimiterX{\frobnorm}[1]{\lVert}{\rVert_{\mathrm{F}}}{#1}
\DeclarePairedDelimiter\frobket{\langle}{\rangle_{\mathrm{F}}}
\DeclarePairedDelimiterX{\abs}[1]{\lvert}{\rvert}{#1}

\newcommand{\iidsim}{\overset{\text{iid}}\sim}
\newcommand{\indsim}{\overset{\text{ind.}}\sim}
\newcommand{\Reals}[1]{\mathbb{R}^{#1}}

\newcommand{\Bern}{\operatorname{Bernoulli}}

\newcommand{\bY}{\mathbf{Y}}
\newcommand{\Y}{\mathbf{Y}}
\newcommand{\by}{\mathbf{y}}
\newcommand{\bX}{\mathbf{X}}
\newcommand{\bD}{\mathbf{D}}
\newcommand{\bL}{\mathbf{L}}
\newcommand{\bR}{\mathbf{R}}
\newcommand{\bP}{\mathbf{P}}
\newcommand{\rmd}{\mathrm{d}}
\newcommand{\bx}{\mathbf{x}}

\newcommand{\bz}{\mathbf{z}}
\newcommand{\bW}{\mathbf{W}}
\newcommand{\bE}{\mathbf{E}}
\newcommand{\bU}{\mathbf{U}}

\newcommand{\bu}{\mathbf{u}}
\newcommand{\bV}{\mathbf{V}}
\newcommand{\bv}{\mathbf{v}}
\newcommand{\bw}{\mathbf{w}}

\newcommand{\bA}{\mathbf{A}}
\newcommand{\btheta}{\boldsymbol\theta}

\newcommand{\bmu}{\boldsymbol\mu}

\newcites{Sup}{References}

\title{\bf \Large Generalized Bayesian Inference for \\ Dynamic Random Dot Product Graphs}

\author{
Joshua Daniel Loyal\\
Department of Statistics, Florida State University
}

\date{}

\begin{document}
\maketitle

\begin{abstract}
    The random dot product graph is a popular model for network data with extensions that accommodate dynamic (time-varying) networks. However, two significant deficiencies exist in the dynamic random dot product graph literature:  (1) no coherent Bayesian way to update one's prior beliefs about the latent positions in dynamic random dot product graphs due to their complicated constraints, and (2) no approach to forecast future networks with meaningful uncertainty quantification. This work proposes a generalized Bayesian framework that addresses these needs using a Gibbs posterior that represents a coherent updating of Bayesian beliefs based on a least-squares loss function. We establish the consistency and contraction rate of this Gibbs posterior under commonly adopted Gaussian random walk priors. For estimation, we develop a fast Gibbs sampler with a time complexity for sampling the latent positions  that is linear in the observed edges in the dynamic network, which is substantially faster than existing exact samplers. Simulations and an application to forecasting international conflicts show that the proposed method's in-sample and forecasting performance outperforms competitors.
\end{abstract}

\noindent
{\it Keywords:} Bayesian Inference; Dynamic Networks; Forecasting; Latent Space Network Models; Statistical Network Analysis.
\noindent

%\spacingset{1.6} % DON'T change the spacing!
\onehalfspacing

\section{Introduction}

Statistical network analysis is a rapidly growing area of statistics with applications in sociology, neuroscience, political science, economics, and others. In general, a network describes the relations, or edges, between pairs of entities, or nodes. Many contemporary scientific inquiries probe how relations between nodes change over time, such as, the formation or dissolution of friendships, the evolution of polarization and partisanship in online social networks~\citep{zhu2023}, and changes in the brain's neural connections due to external stimuli. Often, the appropriate data structure to describe these time-varying phenomena is a time series of networks known as a dynamic network. In this work, we consider dynamic networks on a common set of $n$ nodes observed over $m$ equally spaced discrete time points with edges that change over time. These dynamic networks are represented by a collection of $n \times n$ adjacency matrices $\set{\bY_t}_{t=1}^m$ with entries $\set{Y_{ij,t} \, : \, 1 \leq i,j \leq n, 1 \leq t \leq m}$ called edge variables that quantify the relation between nodes $i$ and $j$ at time $t$. These edge variables can be binary or weighted, taking on any real value, depending on the application.

Existing statistical models for discrete-time dynamic networks can be roughly divided into three categories: autoregressive processes~\citep{almquist2014, almquist2017, sewell2018, jiang2023a, jiang2023b, chang2024}, more general exponential random graph models~\citep{hanneke2010, krivitsky2014stergm}, and latent variable models such as dynamic stochastic block models~\citep{xing2010, yang2011, matias2017} and dynamic latent space models (LSMs)~\citep{sarkar2006, sewell2015}. We refer the reader to \citet{kim2018} for a more complete review. This article focuses on a specific dynamic LSM called a dynamic random dot product graph (RDPG)~\citep{young2007}, which represent each node $i$ at time $t$ with a latent position $\bx_{it}$ in some $d$-dimensional latent space $\mathcal{X} \subseteq \Reals{d}$ such that the expected value of the edge variable between nodes $i$ and $j$ at time $t$, that is, $\mathbb{E}(Y_{ij,t})$, is $\bx_{it}^{\top}\bx_{jt}$. The RDPG has been widely adopted because it is relatively interpretable and flexible enough to contain or approximate several popular network models, such as stochastic block models~\citep{sussman2012} or more general LSMs~\citep{tang2013, tang2025}. We refer the reader to \citet{athreya2018} for an overview of random dot product graphs.

This article addresses the following two problems facing dynamic RDPGs: (1) there is no coherent Bayesian way to update one's prior beliefs about the latent positions, and (2) there is no approach to forecasting future networks with meaningful uncertainty quantification, which is a crucial task in dynamic network analysis.  \citet{xie2020} developed a Bayesian approach for static ($m = 1$) RDPGs; however, extending their method to time-varying latent positions is difficult since it is non-trivial to define a meaningful prior for the latent positions' temporal evolution that respects the constraints often imposed by the latent space. For example, for binary networks, $\mathcal{X} = \{\bx \in \Reals{d} : \norm{\bx}_2 \le 1, \bx \succeq 0\}$, where $\norm{\cdot}_2$ is the Euclidean norm and $\bx \succeq 0$ indicates each element of $\bx$ is greater than or equal to zero. To forecast future networks, existing works~\citep{zhu2016, chenli2018, passino2021} proposed frequentist estimators with heuristic point forecasts. However, most of these methods lack theoretical guarantees for their in-sample performance, and their forecasts all lack confidence bands, which limits their usefulness.

Our key contribution is the introduction of a generalized Bayesian approach~\citep{bissiri2016} with theoretical guarantees that provides a coherent Bayesian framework for updating one's prior beliefs about the latent positions using information in the observed dynamic network based on a generalized posterior distribution. In this work, we adopt the popular belief that the latent positions smoothly vary over time~\citep{sarkar2006, sewell2015, zhu2016, chenli2018, liu2018, zhao2024, macdonald2025}. However, the proposed framework can easily incorporate other priors. To do so, we place hierarchical $r$-th order Gaussian random walk priors on the latent positions with node-specific transition variances. By placing half-Cauchy priors on these transition variances, we show that our Bayesian approach can adaptively pool information to efficiently estimate the latent positions without assuming any common structure shared by all time points~\citep{arroyo2021} or requiring parameter tuning~\citep{macdonald2025}. Moreover, this adaptive pooling allows our approach to estimate the latent positions accurately as the network's sparsity level increases, in contrast to methods that estimate the latent positions separately at each time point~\citep{passino2021, athreya2024}, for which the signal-to-noise ratio in each individual network determines their performance. Lastly, we use these random walk priors to develop a forecast with uncertainty quantification based on a generalized posterior predictive distribution, which is missing in the RDPG literature.

Our last contribution is the development of an efficient Markov chain Monte Carlo (MCMC) algorithm to sample from the generalized posterior distribution with a runtime for sampling the latent positions that scales as the number of observed edges in the dynamic network. This time complexity is substantially faster than existing samplers for discrete-time dynamic LSMs~\citep{loyal2023,zhao2024}, which have runtimes that scale with the number of possible edges in the dynamic network, that is, $O(mn^2)$. As such, our sampler significantly reduced the computational burden of posterior inference for dynamic RDPGs, which alleviates a deterrent to adopting Bayesian inference for large-scale dynamic networks.

The paper is organized as follows. Section~\ref{sec:rdpg_prelim} reviews the dynamic RDPG model and existing estimation schemes. In Section~\ref{sec:gb-dase}, we propose a generalized Bayesian framework for dynamic RDPG inference and introduce our method for forecasting future networks. Section~\ref{sec:theory} provides theoretical guarantees for the recovery of the latent positions. Section~\ref{sec:estimation} introduces an efficient Gibbs sampler for posterior inference. Section~\ref{sec:simulation} presents a simulation study, and we apply the methods to a real dynamic network in Section~\ref{sec:real_data}. Section~\ref{sec:discussion} concludes with a discussion. The supplementary material contains all proofs and additional technical details. A Python package for the proposed method is available at \url{https://github.com/joshloyal/DynamicRDPG}.

\section{Preliminaries}\label{sec:rdpg_prelim}

\subsection{Notation}

Throughout this article, $\frobnorm{\cdot}$ and $\norm{\cdot}_2$ denote the Frobenius and Euclidean norms.  We use $\indsim$ and $\iidsim$ to denote independently distributed and independently and identically distributed. $\mathcal{O}_d$ is the group of $d$-dimensional orthogonal matrices. We use $\text{vec}(\mathbf{A})$ to denote vectorization of a matrix. For square matrices $\bA_1, \dots, \bA_k$, we use $\diag(\bA_1, \dots, \bA_k)$ to denote a block diagonal matrix with $\bA_1, \dots, \bA_k$ as the diagonal blocks. For positive sequences $\set{a_n}$ and $\set{b_n}$, $a_n = O(b_n)$ if there exists a constant $C > 0$ such that $a_n \leq C b_n$ for $n$ large enough.

\subsection{Background and Motivation}\label{sec:rdpg}

Consider a dynamic network represented as a collection of time-index $n \times n$ symmetric hollow adjacency matrices $\set{\bY_t}_{t=1}^m$ with random entries $Y_{ij,t} \in \Reals{}$ and diagonal entries equal to zero. In many settings, the parameters of interest are the means of these adjacency matrices, that is, $\set{\mathbb{E}_0(\bY_t)}_{t=1}^m$, where $\mathbb{E}_0$ denotes the expectation under the true data-generating process. The dynamic RDPG model assumes that these mean matrices are positive semi-definite and low-rank, that is, $\mathbb{E}_0(\bY_t) = \bX_{0t} \bX_{0t}^{\top}$ for $\bX_{0t} \in \Reals{n \times d}$  for $t = 1, \dots, m$ and $d \ll n$. Expressing $\bX_{0t} = (\bx_{01t}, \dots, \bx_{0nt})^{\top} \in \Reals{n \times d}$, the rows $\bx_{01t}, \dots, \bx_{0nt} \in \mathcal{X} \subset \Reals{d}$ represent the latent positions of the $n$ nodes in the network at time $t$ with the property that $\mathbb{E}_0(Y_{ij,t}) = \bx_{0it}^{\top}\bx_{0jt}$. 

The foundation of our generalized Bayesian framework is the following property of dynamic RDPGs. Under the dynamic random dot product graph model, the latent position matrices $\set{\bX_{0t}}_{t=1}^m$ minimize the following expected loss: 
\begin{equation}\label{eq:expect_loss}
     \mathbb{E}_{0}\left(\frobnorm{\bY_t - \bX_t \bX_t^{\top}}^2\right)  \qquad (t = 1, \dots, m),
\end{equation}
or equivalently
\[
    \set{\bX_{0t}}_{t=1}^m \in \argmin_{\bX_1, \dots, \bX_m \in \Reals{n \times d}} \mathbb{E}_0\left(\sum_{t=1}^m\frobnorm{\bY_t - \bX_t \bX_t^{\top}}^2\right). 
\]
The latent position matrices $\set{\bX_{0t}}_{t=1}^m$ are not the unique minimizer of Equation~(\ref{eq:expect_loss}). In fact, for any collection of orthogonal matrices $\set{\bW_t \, : \, \bW_t \in \mathcal{O}_d}_{t=1}^m$, we have that $\set{\bX_{0t} \bW_t}_{t=1}^m$ also minimizes (\ref{eq:expect_loss}) because $(\bX_{0t} \bW_t)(\bX_{0t} \bW_t)^{\top} = \bX_{0t}\bX_{0t}^{\top}$ for any $t = 1, \dots, m$. This non-identifiability occurs in most latent space models~\citep{hoff2002} and implies that the latent positions can only be recovered up to an orthogonal transformation. 

In view of Equation~(\ref{eq:expect_loss}), \citet{sussman2012} motivated the adjacency spectral embedding (ASE) estimator of the latent positions as the minimizer of the empirical loss, that is,
\begin{equation}\label{eq:ase}
    \hat{\bX}_{t}^{(\textrm{ASE})} = \argmin_{\bX_t \in \Reals{n \times d}} \frobnorm{\bY_t - \bX_t \bX_t^{\top}}^2 \qquad (t = 1, \dots, m).
\end{equation}
Under further conditions on the true data-generating process, e.g., the edge variables are independent sub-Gaussian random variables, the ASE estimator is consistent and asymptotically normal~\citep{sussman2014, athreya2015}. 

An issue with the procedure in (\ref{eq:expect_loss})--(\ref{eq:ase}) is that it does not account for any prior beliefs about the possible temporal relationships between the latent position matrices, e.g., smoothly varying. As a result, existing works strictly enforce such beliefs by either imposing shared structure~\citep{arroyo2021} or adding a roughness penalty~\citep{zhu2016, chen2018, liu2018} to Equation~(\ref{eq:ase}). However, these approaches do not allow the data to inform the suitability of the assumed temporal relationship. As such, we introduce a generalized posterior distribution that represents a rational Bayesian update of one's prior beliefs about the latent positions based on the information in the data provided by Equation~(\ref{eq:expect_loss}).

\section{Generalized Bayes for Dynamic RDPGs}\label{sec:gb-dase}

\subsection{Generalized Bayes Dynamic Adjacency Spectral Embedding}

To develop a coherent Bayesian procedure for updating our prior beliefs, we adapt the reasoning of \citet{bissiri2016} to the dynamic RDPG setting. In this framework, the target of inference is an optimal and unknown set of rank $d$ latent position matrices defined as a minimizer of the expected loss function, namely
\begin{equation*}
    \set{\bX_{t}^{(\text{opt})}}_{t=1}^m \in \argmin_{\bX_1, \dots, \bX_m \in \Reals{n \times d}} \mathbb{E}_0\{\ell(\bY_{1:m}; \, \bX_{1:m})\} = \argmin_{\bX_1, \dots, \bX_m \in \Reals{n \times d}} \mathbb{E}_0\left(\sum_{t=1}^m \frobnorm{\bY_t - \bX_t \bX_t^{\top}}^2 \right).
\end{equation*}
Unlike in Section~\ref{sec:rdpg}, the definition of $\bX_{1:m}^{(\text{opt})}$ does not require the existence of a true set of latent positions. Even when each $\mathbb{E}_0(\bY_t)$ is full-rank, it is useful to infer the aspects of the network captured by such a low-rank embedding. However, when the networks come from a dynamic RDPG, then $\bX_{1:m}^{(\text{opt})}$ equals $\set{\bX_{0t}}_{t=1}^m$ up to a collection of orthogonal transformations. 

Since the optimal latent positions cannot be computed without knowing the true data-generating process, we instead rely on the ``best'' conditional distribution for quantifying our subjective beliefs about $\bX_{1:m}^{(\text{opt})}$ to inform us about the optimal latent positions. Let $\Pi(\bX_{1:m})$ denote a prior distribution for the latent positions, $\mathscr{P}$ be the space of conditional probability distributions for the latent positions given the data, and $\mathscr{L}$ be a loss function defined on that space. For any $\nu_1, \nu_2 \in \mathscr{P}$, the loss $\mathscr{L}$ describes whether $\nu_1$ is a better candidate than $\nu_2$ for representing one's posterior belief about $\bX_{1:m}^{(\text{opt})}$.  For coherence of the posterior, \citet{bissiri2016} argued that the loss $\mathscr{L}$ must take the form
\begin{equation}\label{eq:loss_func}
    \mathscr{L}\{\nu(\bX_{1:m} \mid \bY_{1:m})\} =  \lambda \, \mathbb{E}_{\nu}\left\{ \sum_{t=1}^m \frobnorm{\bY_t - \bX_t \bX_t^{\top}}^2\right\} + D_{\textrm{KL}}\{\nu(\bX_{1:m} \mid \bY_{1:m}) \mid \mid  \Pi(\bX_{1:m})\},
\end{equation}
where the expectation is taken with respect to $\nu(\bX_{1:m} \mid \bY_{1:m})$, $D_{\textrm{KL}}(p \mid \mid q)$ denotes the Kullback-Leibler divergence between densities $p$ and $q$, and the parameter $\lambda > 0$ is called the learning rate. The generalized posterior distribution is defined as the optimal conditional distribution $\Pi_{\lambda}(\bX_{1:m} \mid \bY_{1:m}) = \argmin_{\nu \in \mathscr{P}} \mathscr{L}\{\nu(\bX_{1:m} \mid \bY_{1:m})\}$. The unique minimizer is
\begin{equation}\label{eq:gb-dase}
    \Pi_{\lambda}(\bX_{1:m} \mid \bY_{1:m}) \propto \exp\left(-\lambda \sum_{t=1}^m \frobnorm*{\bY_t - \bX_t \bX_t^{\top}}^2 \right) \Pi(\bX_{1:m}),
\end{equation}
which is also known as a Gibbs posterior~\citep{zhang2006,martin2022}. When $\Pi(\bX_{1:m})$ is proper, this Gibbs posterior is a well-defined probability density function because the normalizing constant satisfies $0 < \int \exp\left(-\lambda \sum_{t=1}^m \frobnorm*{\bY_t - \bX_t \bX_t^{\top}}^2 \right) \Pi(\rmd \bX_{1:m}) \leq 1$. 

In addition, the Gibbs posterior provides a principled way to inject prior information into the adjacency spectral embedding, where the learning rate $\lambda$ controls the prior's influence. To clarify, we observe that the loss in (\ref{eq:loss_func}) balances two components: the discrepancy between the posterior and the data $\mathbb{E}_{\nu}\set{\ell(\bX_{1:m}; \bY_{1:m})}$, and the closeness of the posterior to the prior $D_{\text{KL}}\set{\nu(\bX_{1:m} \mid \bY_{1:m}) \mid \mid \Pi(\bX_{1:m})}$.  The learning rate $\lambda$ determines the influence of the first component. As $\lambda \rightarrow 0$, the closeness to the observed networks is not penalized, and one obtains $\Pi_0(\bX_{1:m} \mid \bY_{1:m}) = \Pi(\bX_{1:m})$. In contrast, as $\lambda \rightarrow \infty$, the effect of the prior is negligible and the Gibbs posterior collapses to a point mass centered at the adjacency spectral embedding, that is, $\delta_{\hat{\bX}_{1:m}^{(\text{ASE})}}$. In this way, the Gibbs posterior combines the ASE with one's prior beliefs about the latent positions. As such, we refer to the Gibbs posterior in (\ref{eq:gb-dase}) as the {\it generalized Bayes dynamic adjacency spectral embedding} (GB-DASE). 

In summary, GB-DASE should be interpreted like a standard Bayesian posterior, i.e., quantifying a rational update of one's prior beliefs based on data, with the crucial difference that the target of inference is $\bX_{1:m}^{(\text{opt})}$ instead of some true low-rank factorizations of $\set{\mathbb{E}_0(\bY_t)}_{t=1}^m$. Since the number of dimensions $d$ defines $\bX_{1:m}^{(\text{opt})}$, we believe that $d$ should be subjectively chosen as part of $\ell(\bX_{1:m}; \bY_{1:m})$, rather than inferred from the data. In this context, the optimal latent position matrices represent the most important information in the networks expressible by a $d$-dimensional latent space. For example, one might choose $d = 2$ or $3$ to visualize the most salient features of the dynamic network easily. Alternatively, $d$ can be selected before posterior inference using standard approaches such as the ``elbow'' method~\citep{zhu2006, athreya2015}, which we use in Section~\ref{sec:real_data}.

\subsection{Smoothness Promoting Priors}\label{subsec:prior}

A benefit of the generalized Bayesian framework is that it allows us to incorporate and update our prior beliefs about the temporal behavior of the latent position matrices. In this work, we operate under the common belief that the latent positions smoothly vary over time~\citep{sarkar2006, sewell2015, zhu2016, chenli2018, liu2018, zhao2024}, although one could apply other prior beliefs. More specifically, we define the latent trajectory of node $i$ as the matrix $\bx_{i,1:m} = (\bx_{i1}, \dots, \bx_{im}) \in \Reals{d \times m}$, and our prior belief is that the adjacent columns of this matrix are similar.

To formally quantify this temporal smoothness belief, for $i = 1, \dots, n$, we place independent $r$-th order Gaussian random walk priors on the latent trajectories, that is,
\begin{equation}\label{eq:rw_prior}
    \Delta^r \bx_{it} = \sigma_i {\bw}_{it}, \quad \bw_{it} \iidsim N(\mathbf{0}_d, \mathbf{I}_d) \qquad (t = r+1, \dots, m),
\end{equation}
so that, for $t = r+1, \dots, m$,  the prior density of $\bX_t$ conditioned  $\bX_{1:(t-1)}$ is 
\begin{equation}%\label{eq:x_prior}
    \Pi_{\text{RW}(r)}(\bX_t \mid \bX_{t-1}, \dots, \bX_{t-r}, \sigma_{1:n}) = \prod_{i=1}^n \left(2\pi \sigma_i^2\right)^{-d/2}\exp\left(-\frac{1}{2\sigma_i^2} \norm{\Delta^r \bx_{it}}_2^2\right).
\end{equation}
Here, $\mathbf{I}_d$ is the $d$-dimensional identity matrix, $\mathbf{0}_d$ is a $d$-dimensional vector of ones, $\Delta^r$ denotes the component-wise $r$-th order difference operator, and $\sigma_i > 0$ captures node-specific variability for $i = 1, \dots, n$. This process, which we call a RW($r)$ prior, is also known as an integrated random walk~\citep{young1991} and its' temporal smoothness increases with $r$. 

%An issue with the RW($r$) prior is that it leads to an improper prior on $\bX_{1:m}$, which means that the corresponding generalized posterior may not exist. However, we have the following result that the generalized posterior is proper under the RW($r$) prior. The result follows from the fact that the loss function $\ell(\bX_{1:m}; \bY_{1:m})$ includes the adjacency matrices' zero-valued diagonal elements, which effectively imposes a proper prior on the latent positions.
%\begin{proposition}\label{prop:proper}
%    For any $\lambda > 0$, $r \geq 1$, and $\sigma_1, \dots, \sigma_n > 0$, the Gibbs posterior in Equation~(\ref{eq:gb-dase}) under the improper $r$-th order random-walk priors in Equations~(\ref{eq:rw_prior})--(\ref{eq:x_prior}) is proper.
%\end{proposition}

In this work, we focus on two special cases of the prior in (\ref{eq:rw_prior}) with $r$ set to one or two. When $r = 1$, we recover the first-order random walk prior widely used in the dynamic LSM literature~\citep{sewell2015, zhao2024}, that is, $\bx_{it} = \bx_{i(t-1)} + \sigma_i \bw_{it}$ for $t = 2, \dots, m$. However, setting $r = 2$, we recover a local linear trend model $\bx_{it} = 2 \bx_{i(t-1)} - \bx_{i(t-2)} + \sigma_i \bw_{it}$ for $t = 3, \dots, m$, which has not appeared in the dynamic LSM literature. One's prior belief about the long-term smoothness of the latent trajectories and the forecast function's form should guide the choice among these two priors, as shown in the next paragraph.

\begin{figure}[tb]
\centering \includegraphics[width=\textwidth, keepaspectratio]{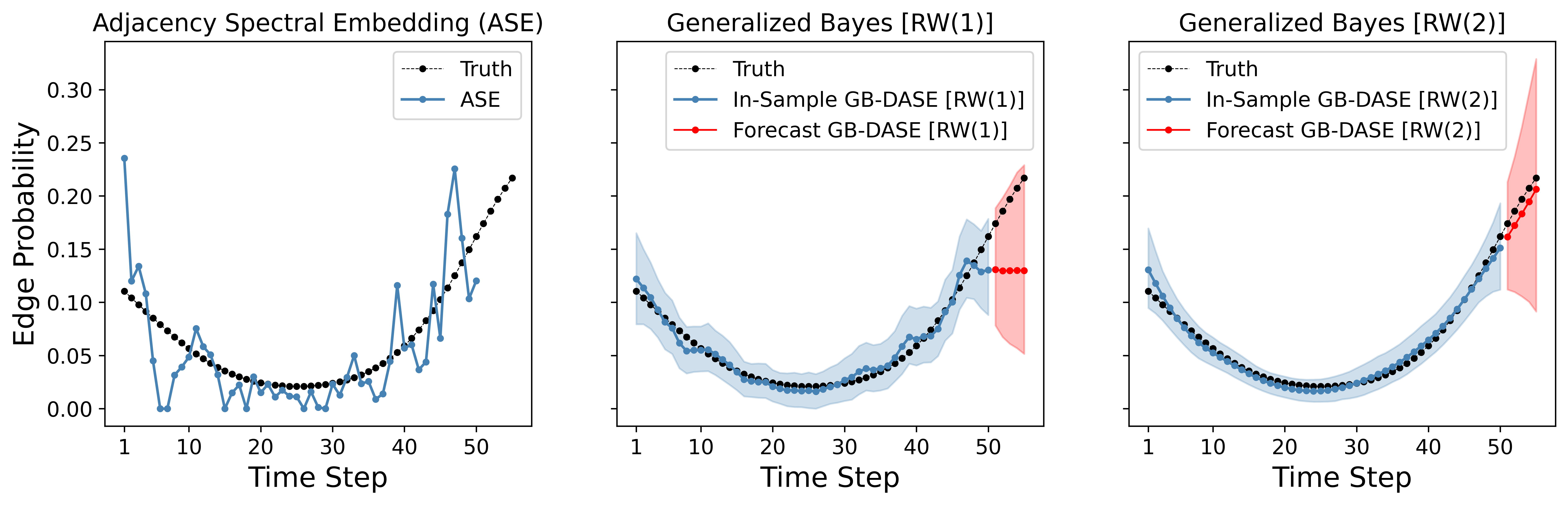}
    \caption{Comparison of the estimated edge probabilities based on (Left) ASE, (Center) GB-DASE with an RW($1$) prior, and (Right) GB-DASE with an RW($2$) prior. The dynamic network has $n = 200$ nodes and $m = 50$ in-sample time points. Five future time points are displayed that were not used to fit the models. The blue points denote in-sample edge probabilities estimated using the three methods. The red points denote the forecasted edge probabilities under GB-DASE using the method in Section~\ref{sec:forecasting}. The shaded regions indicate 95\% pointwise credible intervals and the black points indicate the true edge probabilities.}
    \label{fig:rw_comp}
\end{figure}

Figure~\ref{fig:rw_comp} compares the in-sample estimates and out-of-sample forecasts of the edge probability between two nodes from a binary network, that is, $\mathbb{E}_0(Y_{ij,t})$, using GB-DASE with RW($1$) and RW($2$) priors estimated using Algorithm~\ref{algo:gibbs} in Section~\ref{sec:estimation} and the forecasting procedure described in Section~\ref{sec:forecasting}. Overall, the GB-DASE estimates produce a smoother trajectory for the edge probability than the ASE estimates defined in Equation (\ref{eq:ase}). The estimates from GB-DASE under the RW($1$) prior are essentially a smoothed version of the ASE estimates, with ASE's spurious large fluctuations replaced with smaller bumps. This behavior underscores the generalized posterior's ability to combine ASE with smoothness-promoting priors. Furthermore, GB-DASE under an RW($2$) prior produces a continuous-looking curve. In general, the RW($2$) prior is suitable for modeling trajectories with long-term smooth variation, while the RW($1$) prior produces trajectories with shorter-term, less smooth variation~\citep{young2011}. Another crucial distinction between the two priors is how they extrapolate the edge probability forward in time using the forecasting method proposed in Section~\ref{sec:forecasting}. GB-DASE under an RW($1$) prior produces a constant forecast function, and GB-DASE under an RW($2$) prior produces a linear forecast function. Both can be reasonable forecasts depending on the dynamics of the system.

Another way to understand the behavior of GB-DASE under the proposed RW($r$) prior is to consider the maximum a posteriori (MAP) estimator, that is,
\[
    \hat{\bX}_{1:m}^{(\text{MAP})} = \argmin_{\bX_1, \dots, \bX_m \in \Reals{n \times d}} \left\{ \lambda \sum_{t=1}^m \frobnorm{\bY_t - \bX_t \bX_t^{\top}}^2  + (1/2)\sum_{i=1}^n \sum_{t=r+1}^m \sigma_i^{-2} \norm{\Delta^r \bx_{it}}_2^2\right\}.
\]
 Similar to many dynamic network embeddings, this estimator balances a cross-sectional cost that measures the goodness-of-fit to the observed network time-series with a longitudinal cost that penalizes temporal variation~\citep{zhu2016, chenli2018, liu2018}. However, a difficulty with existing cost-based embeddings is appropriately specifying the trade-off between the two costs, which is governed here by $\lambda$ and $\sigma_{1:n}$. 

 Under the generalized Bayesian framework, a natural solution to allow the generalized posterior to adapt to the underlying temporal variation in the data is to place priors over the node-wise variances. Here, we place independent half-Cauchy priors over the variances, that is, $\sigma_1, \dots, \sigma_n \iidsim C^{+}(0, 1)$, which has been widely advocated as a prior for variances in hierarchical models~\citep{gelman2006}. Furthermore, we show that these half-Cauchy priors allow the posterior to adapt to the latent trajectories' true variability in Section~\ref{sec:theory}. 

The improper RW($r$) prior may result in an improper generalized posterior. As such, we give the initial latent positions independent mean-zero Gaussian priors with identical variances $\sigma_0^2$ to ensure the priors are proper. In summary, the prior on $\bX_{1:m}$ has the density
\begin{align}\label{eq:x_prior}
    \Pi_{\text{RW}(r)}(&\bX_{1:m} \mid \sigma_{1:n}, \sigma_0) = \prod_{i=1}^n \prod_{t=1}^r (2\pi\sigma_0^2)^{-d/2} \exp\left(-\frac{1}{2\sigma_0^2} \norm{\bx_{it}}_2^2 \right) \times \nonumber \\
    &\qquad\qquad\qquad\qquad \prod_{t=r+1}^m \Pi_{\text{RW}(r)}(\bX_t \mid \bX_{t-1}, \dots, \bX_{t-r}, \sigma_{1:n}) \nonumber \\
     &\propto \prod_{i=1}^n \exp\left[-\frac{1}{2} \textrm{vec}(\bx_{i,1:m})^{\top} \left\{\left(\frac{1}{\sigma_i^2} \bD_{r}^{\top} \bD_{r} + \frac{1}{\sigma_0^2} \sum_{s=1}^r \mathbf{e}_s \mathbf{e}_s^{\top}\right) \otimes \mathbf{I}_d\right\} \textrm{vec}(\bx_{i,1:m})\right], %\label{eq:x_prior} \\
\end{align}
where $\mathbf{D}_r$ is a $(m-r) \times m$ matrix representing the $r$-th order finite difference operator acting on $\bv \in \Reals{m}$ and $\mathbf{e}_s \in \Reals{m}$ is the $s$-th standard basis vector.

\subsection{Forecasting Future Networks}\label{sec:forecasting}

Forecasting future edges based on past data is a fundamental task in dynamic network analysis. Formally, our goal is to define a meaningful $k$-step ahead forecast for the network at time $m+k$, that is, $\bY_{m+k}$, for $k \geq 1$. Under the generalized Bayesian paradigm, a strategy to obtain $k$-step ahead forecasts is to rely on a $k$-step ahead predictive distribution constructed based on the generalized posterior~\citep{wu2021, loaiza2021, frazier2024}. However, unlike traditional Bayes, since we have not specified a complete statistical model for the networks, it is not clear how to define a meaningful predictive distribution. As such, to proceed, we assume the following general statistical model for the future networks: $\bY_{m+k}$ comes from an RDPG with  latent position matrix $\bX_{m+k} \in \Reals{n \times d}$, that is, $\bY_{m+k}$ has a marginal density $p_0(\bY_{m+k} \mid \bX_{m+k})$ such that $\mathbb{E}_0(\bY_{m+k}) = \bX_{m+k}\bX_{m+k}^{\top}$. 

Next, we define the {\it $k$-step ahead RW(r)-generalized  posterior predictive distribution} as 
\begin{align}\label{eq:gen_post_pred}
    Q_{\lambda}^{(r)}(&\bY_{m+k} \mid \bY_{1:m}) = \nonumber \\
    &\quad\int p_0(\bY_{m+k} \mid \bX_{m+k}) \Pi_{\text{RW}(r)}(\rmd \bX_{m+k} \mid \bX_{1:m}, \sigma_{1:n}) \Pi_{\lambda}^{(r)}(\rmd \bX_{1:m}, \rmd \sigma_{1:n} \mid \bY_{1:m}).
\end{align}
Here, $\Pi_{\lambda}^{(r)}(\bX_{1:m}, \sigma_{1:n} \mid \bY_{1:m})$ is the generalized posterior under the $r$-th order random walk and half-Cauchy priors defined in Section~\ref{subsec:prior}, and $\Pi_{\text{RW}(r)}(\bX_{m+k} \mid \bX_{1:m}, \sigma_{1:n})$ is the marginal density of $\bX_{m+k}$ under an $r$-th order Gaussian random walk with variances $\sigma_{1:n}$ propagated $k$-steps forward in time conditioned on the initial states $\bX_{1:m}$. Intuitively, the RW($r)$-generalized posterior predictive distribution breaks down the uncertainty in $\bY_{m+k}$ into three terms: (1) the uncertainty in $\bX_{1:m}$ and $\sigma_{1:n}$ given the observed networks captured by the generalized posterior $\Pi_{\lambda}^{(r)}(\bX_{1:m}, \sigma_{1:n} \mid \bY_{1:m})$, (2) the uncertainty in the future latent positions $\bX_{m+k}$ given $\bX_{1:m}$ and $\sigma_{1:n}$ represented by $\Pi_{\text{RW}(r)}(\bX_{m+k} \mid \bX_{1:m}, \sigma_{1:n})$, and (3) the uncertainty in the future network $\bY_{m+k}$ given $\bX_{m+k}$ defined by $p_0(\bY_{m+k} \mid \bX_{m+k})$.

Our strategy to obtain $k$-step ahead forecasts is to use the expectation of the $k$-step ahead RW($r$)-generalized posterior predictive distribution, that is, $\mathbb{E}_{Q_{\lambda}^{(r)}}(\bY_{m+k} \mid \Y_{1:m})$. To compute this forecasting function, we re-express Equation~(\ref{eq:gen_post_pred}) as follows:
\[
    Q_{\lambda}^{(r)}(\bY_{m+k} \mid \bY_{1:m}) = \mathbb{E}_{P_{\lambda}^{(r)}}\left\{p_0(\bY_{m+k} \mid \bX_{m+k}) \mid \bY_{1:m}\right\},
\]
where 
\begin{equation}\label{eq:lp_gen_post}
    P_{\lambda}^{(r)}(\bX_{m+k} \mid \bY_{1:m}) = \mathbb{E}_{\Pi_{\lambda}^{(r)}}\{\Pi_{\text{RW}(r)}(\bX_{m+k} \mid \bX_{1:m}, \sigma_{1:m}) \mid \bY_{1:m}\},
\end{equation}
which we call the {\it $k$-step ahead RW(r)-generalized posterior density of $\bX_{m+k}$}. It follows from Fubini's theorem that
\begin{align}
    \mathbb{E}_{Q_{\lambda}^{(r)}}(\bY_{m+k} \mid \bY_{1:m}) &= \int \bY_{m+k} \mathbb{E}_{P_{\lambda}^{(r)}}\left\{p_0(\bY_{m+k} \mid \bX_{m+k}) \mid \bY_{1:m}\right\} \rmd \bY_{m+k} \nonumber \\
    &=  \mathbb{E}_{P_{\lambda}^{(r)}}\left\{\int \bY_{m+k} \, p_0(\bY_{m+k} \mid \bX_{m+k}) \rmd \bY_{m+k} \ \middle| \ \bY_{1:m}\right\}  \nonumber \\
    &= \mathbb{E}_{P_{\lambda}^{(r)}}(\bX_{m+k} \bX_{m+k}^{\top} \mid \bY_{1:m}) \label{eq:gen_forecast}
\end{align}
under the assumption that $\mathbb{E}_0(\bY_{m+k}) = \bX_{m+k} \bX_{m+k}^{\top}$. As such, we can compute forecasts without specifying the density $p_0(\bY_{m+k} \mid \bX_{m+k})$. In addition, we can quantify the uncertainty in this forecast by constructing credible intervals based on $P_{\lambda}^{(r)}(\bX_{m+k} \mid \bY_{1:m})$.

\subsection{Inferring the Learning Rate $\lambda$}\label{subsec:lambda}

A barrier to the adoption of generalized Bayesian inference is the choice of the learning rate $\lambda$. Here, we demonstrate how to infer $\lambda$ based on the fact that, under an appropriate reparameterization,  $\lambda$ acts as the precision parameter of the data under a Gaussian pseudolikelihood. Following \citet{bissiri2016}, we incorporate $\lambda$ into the loss function and specify a hierarchical prior $\Pi(\lambda)$. Specifically, we define the augmented loss function $\ell(\bX_{1:m}, \, \lambda;\, \bY_{1:m}) = \lambda \sum_{t=1}^m \frobnorm{\bY_m - \bX_m\bX_m^{\top}}^2 - \xi \log(\lambda)$ resulting in the follow Gibbs posterior:
\[
    \Pi_{\tilde{\lambda}}^{(r)}(\bX_{1:m}, \sigma_{1:n}, \lambda \mid \bY_{1:m}) \propto \exp\left\{-\tilde{\lambda} \ell(\bX_{1:m}, \, \lambda; \, \bY_{1:m})\right\} \Pi(\lambda) \Pi_{\text{RW}(r)}(\bX_{1:m} \mid \sigma_{1:n}, \sigma_0) \Pi(\sigma_{1:n}),
\]
for some $\tilde{\lambda} > 0$ and $\xi \geq 0$. \citet{bissiri2016} observed that the subjective choice of $\tilde{\lambda}$ and $\xi$ is often easier than that of $\lambda$.

In this case, a natural default choice is to set $\tilde{\lambda} = 1/4$ and $\xi = n(n+1) m$, so that
\begin{align}\label{eq:gauss_gibbs_post} 
    &\exp\bigg\{-\tilde{\lambda} \ell(\bX_{1:m}, \, \lambda; \bY_{1:m})\bigg\} = \lambda^{n(n+1)m/4} \exp\bigg\{-\frac{\lambda}{4} \sum_{t=1}^m \frobnorm{\bY_m - \bX_m \bX_m^{\top}}^2\bigg\} \nonumber \\
    &=\lambda^{{n(n-1)}m/4} \exp\left\{-\frac{\lambda}{2} \sum_{t=1}^m \sum_{i < j} (Y_{ij,t} - \bx_{it}^{\top}\bx_{jt})^2\right\} \times \nonumber \\
    &\qquad\qquad (\lambda/2)^{nm/2} \exp\left\{-\frac{(\lambda/2)}{2}\sum_{i=1}^n \text{vec}(\bx_{i,1:m})^{\top}\text{vec}(\bx_{i,1:m})\right\},
\end{align}
where we used the fact that the adjacency matrices are hollow. In this way, the Gibbs posterior can be seen as a Bayesian update of a Gaussian likelihood with precision parameter $\lambda$ for the adjacency matrices' off-diagonal elements and $\lambda/2$ for the diagonal elements. In other words, the Gibbs posterior is an ``approximate'' Bayesian model under a Gaussian pseudolikelihood. Other choices for $\tilde{\lambda}$ and $\xi$ are possible. For example, setting $\tilde{\lambda} = \alpha /4$ and $\xi = n(n+1)m$ for $0 < \alpha < 1$ results in a fractional posterior~\citep{walker2001, bhattacharya2019} under a Gaussian pseudolikelihood, which could deflate the effect of the misspecified likelihood. To infer $\lambda$, we set its prior $\Pi(\lambda) = \textrm{Gamma}(a_{\lambda}, b_{\lambda})$, which leads to a simple Gibbs sampling step outlined in Section~\ref{sec:estimation}. 

%Lastly, we note that in the context of static ($m = 1$) RDPGs, \citet{xie2020} called posterior inference based on a similar Gaussian pseudolikelihood with $\lambda = 1$ {\it Gaussian spectral embedding}. However, our initial experiments found that fixing $\lambda = 1$ performed poorly on simulated and real data compared to the traditional ASE estimator, especially when the networks are sparse causing the edge-variable's variance to be high. This observation agrees with the simulation studies in \citet{xie2020} and our interpretation of $\lambda$ as a precision parameter. As such, we advocate for either inferring $\lambda$ or setting it's value based on its role as a precision parameter. Regardless, we stress that although we use this connection with a Gaussian likelihood to choose $\lambda$, unlike Gaussian spectral embedding, we rely on the Gibbs posterior justification, so the generalized posterior $\Pi_{\tilde{\lambda}}(\bX_{1:m}, \sigma_{1:n}, \lambda \mid \bY_{1:m})$ should be regarded as a proper update of beliefs about $\bX_{1:m}^{(\text{opt})}$ rather than an approximate model.

\section{Theoretical Guarantees}\label{sec:theory}

Here, we establish the consistency and contraction rate of generalized Bayesian dynamic adjacency spectral embedding under first-order Gaussian random walk priors. Formally, our theoretical results concern the frequentist properties of the following Gibbs posterior distribution for $\bX_{1:m}$ with a fixed $\lambda > 0$:
\begin{align*}
    \Pi_{\lambda}(\bX_{1:m} \in \mathcal{A} \mid \bY_{1:m}) = \frac{N_{nm}(\mathcal{A})}{D_{nm}}
\end{align*}
where
\begin{align*}
    N_{nm}(\mathcal{A}) = \int_{\mathcal{A}} \prod_{t=1}^m \frac{e^{-\lambda \frobnorm{\bY_t - \bX_t \bX_t^{\top}}^2}}{e^{-\lambda  \frobnorm{\bY_t - \bX_{0t}\bX_{0t}^{\top}}^2}} \Pi_{\text{RW}(1)}(\rmd\bX_{1:m} \mid \sigma_0)\quad  \text{and} \quad D_{nm} = N_{nm}(\otimes_{t=1}^m \Reals{n \times d}),
\end{align*}
for any measurable set $\mathcal{A} \subset \otimes_{t=1}^m \Reals{n \times d}$. In the previous definition, we use $\Pi_{\textrm{RW}(1)}(\bX_{1:m} \mid \sigma_0)$ to denote the marginal prior distribution of the latent position matrices under the first-order random walk priors, that is, $\Pi_{\textrm{RW}(1)}(\bX_{1:m} \mid \sigma_0) = \int \Pi_{\textrm{RW}(1)}(\bX_{1:m} \mid \sigma_{1:n}, \sigma_0) \Pi(\rmd \sigma_{1:n})$.

To proceed, we assume that the observed networks are generated from a general low-rank signal matrix plus sub-Gaussian noise matrix model. Namely, that 
\[
    \bY_t = \bX_{0t} \bX_{0t}^{\top} + \bE_t \qquad (t = 1, \dots, m),
\]
where $\bX_{01}, \dots, \bX_{0m} \in \Reals{n \times d}$ and the noise matrices $\bE_1, \dots, \bE_m \in \Reals{n \times n}$ are marginally mean-zero sub-Gaussian matrices with a common variance proxy $\tau > 0$, that is, for all $\mathbf{A} \in \Reals{n \times n}$ with $\frobnorm{\mathbf{A}}^2 = 1$ and all $s > 0$, $\mathbb{P}_0(\abs{\tr(\mathbf{A}^{\top} \mathbf{E}_t)} \geq s) \leq e^{-\tau s^2}$ for $t = 1, \dots, m$. This data-generating process implies that $\mathbb{E}_0(\bY_t) = \bX_{0t} \bX_{0t}^{\top}$ for $1 \leq t \leq m$. Furthermore, this model allows for some forms of cross-sectional and longitudinal edge dependence conditioned on the latent position matrices. In particular, we do not make any assumptions about the noise matrices' temporal dependence; however, this comes at the expense of requiring $m = O\{\log(n)\}$ in the subsequent results. Lastly, this model reduces to an edge-independent Bernoulli RDPG when the elements of the noise matrices are centered Bernoulli random variables, that is, $E_{ij,t} \indsim \Bern(\bx_{0it}^{\top}\bx_{0jt}) - \bx_{0it}^{\top}\bx_{0jt}$ for $1 \leq i \leq j \leq n$ and $1 \leq t \leq m$.

Next, we define an appropriate parameter space for the true latent position matrices that reflects a smooth evolution of the latent trajectories over time. Specifically, we assume the true latent position matrices are in the following pointwise adaptive dependence (PAWD) parameter space introduced in \citet{zhao2024}:
\begin{align*}
    \textrm{PWAD}(\bL) = \bigg\{& \bX_1, \dots, \bX_m \in \Reals{n \times d} \, : \\
    &\max_{1 \leq t \leq m} \norm{\bx_{it}}_2 \leq C_{\bx}, \quad \max_{2 \leq t \leq m } \norm{\bx_{it} - \bx_{i(t-1)}}_2 \leq \frac{dL_i}{nm} \quad (i = 1, \dots, n) \bigg\},
\end{align*}
where $\bL = (L_1, \dots, L_n)^{\top} \in \Reals{n}$ such that $\bL \succeq 0$ and $C_{\bx} > 0$ is a constant independent of $n$ and $m$. The condition that the Euclidean norm of the latent positions is bounded is trivially satisfied by Bernoulli RDPGs since the edge probabilities must lie between zero and one.

The following theorem shows that, under the aforementioned conditions, the generalized posterior distribution is consistent under the average Frobenius norm with a contract rate that adapts to the underlying variation of the latent position matrices. The proof is largely based on the mathematical tools developed in \citet{xie2020} and \citet{zhao2024}.
\begin{theorem}\label{thm:post_con}
    Let $\bY_1, \dots, \bY_m$ be symmetric random $n \times n$ matrices and for each $t \in \set{1, \dots, m}$ let $\mathbb{E}_0(\bY_t) = \bX_{0t}\bX_{0t}^{\top}$ from some $\bX_{0t} \in \Reals{n \times d}$, where $d/n \rightarrow 0$ and $m = O\{\log(n)\}$. Assume that $\set{\bX_{0t}}_{t=1}^m \in \mathrm{PAWD}(\bL)$ and $\bX_{0t}$ is full-rank with minimum singular value $\tilde{\sigma}_{\mathrm{min}}(\bX_{0t}) \geq \tilde{\sigma} \sqrt{n}$ for some constant $\tilde{\sigma} > 0$, for each $t = 1, \dots, m$. Define the contraction rate
    \[
        \epsilon_{nm} = \frac{d^{1/2} \norm{\bL}_{\infty}^{1/3}}{m^{1/3}n^{2/3}} + \sqrt{\frac{d\log nm}{nm}},
    \]
    where $\norm{\bL}_{\infty} = O\{nm \log(nm)\}$. If $\set{\bY_t - \mathbb{E}_0(\bY_t)}_{t=1}^m$ are marginally mean-zero sub-Gaussian random matrices with a common variance proxy $\tau$, then under the first-order random walk prior $\Pi_{\mathrm{RW}(1)}( \bX_{1:m} \mid \sigma_{1:n}, \sigma_0)$ on $\set{\bX_t}_{t=1}^m$ and the half-Cauchy priors on the node-wise variances $\sigma_1, \dots, \sigma_n \iidsim C^{+}(0,1)$, there exists some constants $M_1, M_2 > 0$ and a constant $C_{\tau,\lambda}$ only depending on $\tau, \lambda$, $C_{\bx}$, and $\tilde{\sigma}$, such that 
    \begin{align*}
        \mathbb{E}_0\left\{\Pi_{\lambda}\left(\max_{1 \leq t \leq m} \frobnorm{\bX_t\bX_t^{\top} - \bX_{0t}\bX_{0t}^{\top}}  > M_1  n \epsilon_{nm}  \ \middle| \ \bY_{1:m} \right)\right\} &\leq 2 \exp(-C_{\tau,\lambda} M_1 d n^2 \epsilon_{nm}^2), \\
        \mathbb{E}_0\left\{\Pi_{\lambda}\left(\max_{1 \leq t \leq m} \inf_{\bW_t \in \mathcal{O}_d} \frobnorm{\bX_t - \bX_{0t} \bW_{t}}^2 > M_2  n \epsilon_{nm}^2  \ \middle| \ \bY_{1:m} \right)\right\} &\leq 2 \exp(-C_{\tau,\lambda} M_2 d n^2 \epsilon_{nm}^2)
    \end{align*}
    for sufficiently large $n$ and $m$.
\end{theorem}

Theorem~\ref{thm:post_con} demonstrates that the generalized posterior concentrates its mass around $\set{\mathbb{E}_0(\bY_t)}_{t=1}^m$ and the true latent positions up to a collection of orthogonal transformations. These orthogonal alignments are necessary since the latent positions are only identified up to an orthogonal transformation, as discussed in Section~\ref{sec:rdpg}. Moreover, we require an often assumed minimum singular value condition~\citep{xie2020, ma2020} that $\tilde{\sigma}_{\text{min}}(\bX_{0t}) \geq \tilde{\sigma} \sqrt{n}$ for a positive constant $\tilde{\sigma} > 0$ and $t = 1,\dots,m$. If the true latent positions are viewed as random, this condition holds when the entries of $\bX_{0t}$ are i.i.d.~mean-zero random variables with bounded variances for all $1 \leq t \leq m$ and $d \ll n$.

The contraction rate in Theorem~\ref{thm:post_con} is rate-optimal up to a logarithmic factor for independent-edge Gaussian and Bernoulli dynamic RDPGs; see Theorem~2 of \citet{zhao2024} for the minimax rate for Gaussian RDPGs, but their proof is easily adapted to Bernoulli RDPGs. Moreover, the contraction rate is adaptive to the true latent trajectories' variability since the rate is an increasing function of $\norm{\mathbf{L}}_{\infty}$, implying smoother latent trajectories lead to better rates. However, the rate is at most $\{\log(nm)/nm\}^{1/2}$, which is the minimax rate up to a logarithmic factor for recovering a static set of latent positions given $m$ networks~\citep{xie2020}. Moreover, the rate improves with the number of time points $m$, which demonstrates the generalized posterior's ability to locally pool information across networks. When the latent trajectories are smooth, this behavior motivates using GB-DASE to estimate the latent trajectories instead of using separate ASEs at each time point~\citep{passino2021, athreya2024},  since the latter approach's convergence rate does not improve with $m$.

\section{Computationally Efficient MCMC}\label{sec:estimation}

We propose a simple and efficient Markov chain Monte Carlo (MCMC) algorithm to sample from the Gibbs posterior outlined in Section~\ref{subsec:lambda} under the random walk priors in (\ref{eq:x_prior}). Specifically, we develop a Gibbs sampling scheme using numerical techniques for sparse matrices that can sample the latent position matrices $\bX_{1:m}$ in $O(E)$ time, where $E = \sum_{t=1}^m \sum_{1\leq i < j \leq n} 1\set{y_{ij,t} \neq 0}$ is the number of observed edges in the dynamic network.  This runtime is substantially faster than existing samplers for discrete-time dynamic latent space models, which have runtimes that scale with the number of dyads in the network, that is, $O(mn^2)$. The proposed algorithm is faster because real-world networks are often sparse, meaning $E \ll mn^2$. The Gibbs sampling algorithm is outlined in Algorithm~\ref{algo:gibbs}. Appendix~\ref{sec:mcmc_details} of the supplement contains a derivation of the sampler and details on parameter initialization and post-processing of the posterior samples to account for the latent positions' rotation invariance discussed in Section~\ref{sec:rdpg}. In what follows, we derive the claimed $O(E)$ time complexity of sampling the latent positions, that is, Step~1 in Algorithm~\ref{algo:gibbs}.

\begin{sampler}
    Given initial parameters $\btheta^{(0)} = \{\bX_{1:m}^{(0)}, \sigma_{1:n}^{2 \, (0)}, \nu_{1:n}^{(0)}, \lambda^{(0)}\}$, precompute the following sparse-banded matrices with bandwidths $rd$, $0$, and $d-1$, respectively:
    \[
        \mathbf{K}_1 = \bD_r^{\top} \bD_r \otimes \mathbf{I}_d, \quad \mathbf{K}_2 = (1/\sigma_0^2) \sum_{s=1}^r \mathbf{e}_s \mathbf{e}_s^{\top} \otimes \mathbf{I}_d, \quad \bR = 
        \sum_{i=1}^n \diag(\bx_{i1}^{(0)}\bx_{i1}^{(0) \, \top}, \dots, \bx_{im}^{(0)}\bx_{im}^{(0) \, \top}).
    \]
    For $s = 1 \dots, S$, set $\btheta^{(s)} = \btheta^{(s-1)}$ and then do the following:
    \begin{enumerate}
        \item For each $i = 1, \dots, n$, sample node $i$'s latent trajectory \\
            $\text{vec}(\bx_{i,1:m}) \sim N\{\text{vec}(\bmu_{i,1:m}), \mathbf{P}_i^{-1}\}$ as follows:
        \begin{enumerate}
            \item Calculate the sparse-banded precision matrix 
            \[
                    \bP_i = (1/\sigma_i^{(s)}) \, \mathbf{K}_1 + \mathbf{K}_2 + (\lambda^{(s)}/2) \, \mathbf{I}_{md} + \lambda^{(s)} \, \bR_i 
\]
where 
        \[ 
                    \bR_i = \bR - \diag(\bx_{i1}^{(s)}\bx_{i1}^{(s) \, \top}, \dots, \bx_{im}^{(s)}\bx_{im}^{(s)\, \top}).
        \]
    \item Obtain $\bP_i$'s Cholesky factor $\bL_i$ such that $\bL_i \bL_i^{\top} = \bP_i$. 
    \item Solve \begin{equation}\label{eq:normal_eq}
            \bL_i \bL_i^{\top} \text{vec}(\bmu_{i,1:m}) = \lambda^{(s)} \begin{pmatrix}
                \sum_{j \in \mathcal{N}_{i1}} \bx_{j1}^{(s)} y_{ij,1} \\
        \vdots \\
                \sum_{j \in \mathcal{N}_{im}} \bx_{jm}^{(s)} y_{ij,m}
\end{pmatrix}
    \end{equation}
        by forward-substitution and backward-substitution to obtain $\text{vec}(\bmu_{i,1:m})$.
    \item Sample $\bz_i \sim N(\mathbf{0}_{md}, \mathbf{I}_{md})$ and solve $\bL_i \bu_i = \bz_i$ for $\bu_i$ by back-substitution, so $\bu_i \sim N(\mathbf{0}_{md}, \bP_i^{-1})$. 
    \item Set $\text{vec}(\bx_{i,1:m}^{(s)}) = \text{vec}(\bmu_{i,1:m}) + \bu_i$.%, so that $\text{vec}(\bX_i) \sim N(\text{vec}(\mathbf{M}_i), \mathbf{P}_i^{-1})$.
    \item Update $\mathbf{R} = \mathbf{R}_i + \diag(\bx_{i1}^{(s)}\bx_{i1}^{(s) \, \top}, \dots, \bx_{im}^{(s)}\bx_{im}^{(s) \, \top})$.
        \end{enumerate}
    \item For each $i = 1, \dots, n$, sample the $i$-th node's variance parameters as follows:
        \begin{enumerate}
            \item $\sigma_i^{2 \, (s)} \sim \text{Inverse-Gamma}[\{(m-r)d + 1\}/2, (1/2)\norm{\sum_{t=r}^m\Delta^r \bx_{it}^{(s)})}_2^2 + (1/\nu_i^{(s)})]$.
            \item $\nu_i^{(s)} \sim \text{Inverse-Gamma}\{1, 1 + (1/\sigma_i^{(s)})\}$.
        \end{enumerate}
    \item Sample $\lambda^{(s)} \sim \text{Gamma}(a_{\lambda} + n(n+1)m/4, b_{\lambda} + (1/4) \sum_{t=1}^m \frobnorm{\bY_t - \bX_t^{(s)} \bX_t^{(s) \, \top}}^2)$.
    \end{enumerate}
    \caption{Gibbs sampler for GB-DASE. We set $a_{\lambda} = b_{\lambda} = 10^{-3}$ to induce a broad prior.}
    \label{algo:gibbs}
\end{sampler}

Step 1 in Algorithm~\ref{algo:gibbs} outlines the procedure for sampling each latent trajectory $\bx_{i, 1:m} = (\bx_{i1}, \dots, \bx_{im}) \in \Reals{d \times m}$ efficiently from its Gaussian full-conditional using the sparsity of the precision matrices and the observed network. First, the precision matrix $\mathbf{P}_i$ and the right-hand side of Equation~(\ref{eq:normal_eq}) can be computed in $O(md^2 + d \sum_{t=1}^m \abs{\mathcal{N}_{it}})$ time, where $\mathcal{N}_{it} = \set{j \, : \, y_{ij,t} \neq 0}$ is the neighborhood of node $i$ at time $t$. This time follows from the fact that the precision matrix $\bP_i$ can be calculated in $O(md^2)$ time by appropriate caching and sparse matrix addition, see Step (1.a) and Step (1.f) in Algorithm~\ref{algo:gibbs}. Also, the right side of Equation~(\ref{eq:normal_eq}) explicitly accounts for the network's sparsity and can be computed in $O(d \sum_{t=1}^m \abs{\mathcal{N}_{it}})$ time. 

Next, we use the sparse linear algebra routines proposed in \citet{rue2001} to avoid costly matrix inversions and sample from each latent trajectory's Gaussian full-conditional distribution in $O(md^3)$ time. These steps are outlined in Steps (1.b)--(1.e) of Algorithm~\ref{algo:gibbs}. Note that $\bP_i \in \Reals{md \times md}$ is a banded matrix with bandwidth $rd$, which follows from the fact that $\bP_i$ is a sum of four banded matrices with bandwidths $rd$, $0$, $0$, and $d - 1$, respectively. As such, we perform Step (1.b) in $O(mr^2d^3)$ operations using Algorithm 4.3.5 in \citet{golub1996}, e.g., we use {\tt cholesky\_banded} in SciPy~\citep{scipy2020} version 1.12.0. Furthermore, the Cholesky factor $\bL_i$ is also a banded matrix with lower bandwidth $rd$ by Theorem 4.3.1 in \citet{golub1996}. As such, the system in Step (1.c) can be solved in $O(mrd^2)$ operations using Algorithms 4.3.2 and 4.3.3 in \citet{golub1996}, e.g., we use {\tt cho\_solve\_banded} in SciPy version 1.12.0. Similarly, the back-substitution in Step (1.d) can be performed in $O(mrd^2)$ operations, e.g., we use {\tt spsolve\_triangular} in SciPy version 1.12.0. As such, the overall cost of Step 1 is $O(nmd^3 + nmd^2 + Ed)$, which is $O(E)$ since $nm \leq E$ in most networks, that is, all nodes have at least one edge at all time points. 

\section{Simulation Study}\label{sec:simulation}

%\subsection{Parameter Recovery}\label{sec:recovery}

We present a simulation study comparing GB-DASE with two established ASE-based methods for estimating dynamic RDPGs in terms of recovering the latent trajectories. An additional simulation study comparing the $k$-step ahead forecasting performance of GB-DASE with four ASE-based competitors is in Appendix~\ref{sec:forecast_sim} of the supplement.

The first ASE-based competitor is the traditional ASE applied separately to each adjacency matrix, that is, $\hat{\bX}_{1:m}^{(\text{ASE})}$ defined in Equation~(\ref{eq:ase}). The second ASE-based method is the omnibus embedding (OMNI)~\citep{levin2017}, which finds a separate embedding at each time point $\hat{\bX}_{1:m}^{(\text{OMNI})}$ by applying ASE to the $mn \times mn$ omnibus matrix, that is, the block matrix with the ($t, s$)-th block equal to $(\bY_t + \bY_s)/2$ for $1 \leq t,s \leq m$. We estimated GB-DASE with RW(1) and RW(2) priors using 2,500 samples from Algorithm~\ref{algo:gibbs} after a burn-in of 2,500 iterations and used the posterior means as point estimates. For all estimators, we assume oracle knowledge of the latent space dimension $d$. 

We generate synthetic networks from a Bernoulli dynamic RDPG with independent edges and a latent space dimension $d = 2$, so that
\begin{equation}\label{eq:bern_rdpg}
    Y_{ij,t} \indsim \Bern(\bx_{it}^{\top}\bx_{jt}) \qquad (1 \leq i < j \leq n, \ 1 \leq t \leq m).
\end{equation}
To produce latent trajectories that lie in $\mathcal{X} = \{\bx \in \Reals{d}\, : \, \norm{\bx}_2 \leq 1, \bx \succeq 0\}$, we set $\bx_{it} = \rho \, d^{-1/2} \, \sigma\{\tilde{\bx}_i(t)\}$ for $1 \leq i \leq n$ and $1 \leq t \leq m$, where the scalar $\rho \in [0,1]$ and $\sigma(\cdot)$ is the standard logistic function applied component-wise to the value at time $t$ of a $d$-dimensional vector-valued function $\tilde{\bx}_{i}(s) = (x_{i,1}(s), \dots, x_{i,h}(s))^{\top} : [0,m] \rightarrow \Reals{d}$. In the subsequent experiments, we set $\rho$ to fix the overall expected density of the networks. We sample $\tilde{x}_{i,h}(\cdot) \iidsim \text{GP}(0, C)$ for $1 \leq i \leq n$ and $1 \leq h \leq d$, where $\text{GP}(0, C)$ denotes a mean-zero Gaussian process with covariance function $C(t, t')$. To generate latent trajectories with different levels of smoothness, we use a Mat\'ern covariance function
\[
    C(t,t') =  \frac{a^2}{\Gamma(\nu)2^{\nu - 1}} \left(\frac{\sqrt{2\nu} \abs{t - t'}}{b} \right)^{\nu} K_{\nu}\left(\frac{\sqrt{2\nu} \abs{t - t'}}{b}\right),
\]
where $a,b,\nu > 0$ and $K_{\nu}(\cdot)$ is the modified Bessel function of the second kind, which results in sampled functions that are $\lceil \nu \rceil - 1$ times mean-squared differentiable~\citep{rasmussen2006}. We set the standard deviation $a = \sqrt{5}$ and length scale $b = m/3$, which results in the latent trajectories displayed in Figure~\ref{fig:trajectories_exp} of Appendix~\ref{sec:vis_sim} of the supplement. We evaluate the estimation error using the following root-mean-squared error (RMSE):
\[
    \text{RMSE}_{\bX_{1:m}} = \left\{ \frac{1}{mnd} \sum_{t=1}^m \inf_{\mathbf{W}_t \in \mathcal{O}_d} \frobnorm{\bX_t - \hat{\bX}_t \mathbf{W}_t}^2\right\}^{1/2},
\]
which closely matches the error bounded by Theorem~\ref{thm:post_con}. In all experiments, we recorded the errors averaged over 50 independent replications. Appendix~\ref{sec:vis_sim} of the supplement contains visualizations of the estimated and true latent trajectories for a single replication to provide a more nuanced picture of the performance differences between the estimators.

Figures~\ref{fig:recovery_nodes} and \ref{fig:recovery_times} report the results for networks with expected edge densities of $0.1$, $0.2$, and $0.3$ and relatively smooth ($\nu = 2.5$) and erratic ($\nu = 0.5$) latent trajectories. In Figure~\ref{fig:recovery_nodes}, we vary the number of nodes $n \in \set{25, 50, 100, 200, 400}$ for a fixed $m = 50$. In Figure~\ref{fig:recovery_times}, we vary the number of time points $m\in \set{10, 25, 50, 100, 200}$ for a fixed $n = 50$.  In all settings, the GB-DASE methods perform best. 

\begin{figure}[tbp]
\centering \includegraphics[height=0.35\textheight, width=\textwidth, keepaspectratio]{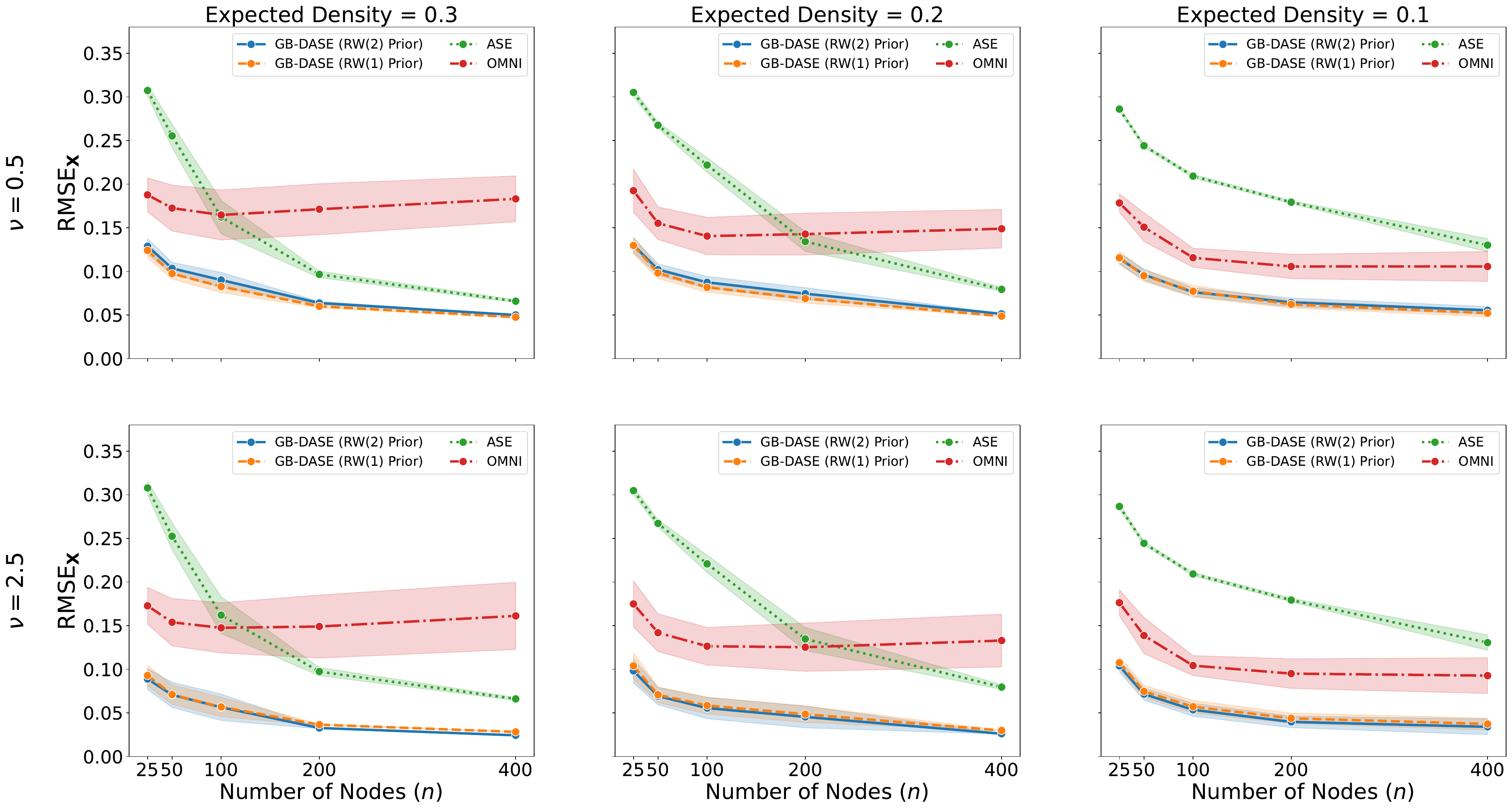}
    \caption{RMSEs for a fixed $m = 50$, varying $n$, expected edge densities, and $\nu = 0.5$ (Top) and $\nu = 2.5$ (Bottom). The curves and shaded regions indicate averages and one standard deviations over 50 independent replicates, respectively.}
    \label{fig:recovery_nodes}
\end{figure}

\begin{figure}[tbp]
\centering \includegraphics[height=0.35\textheight, width=\textwidth, keepaspectratio]{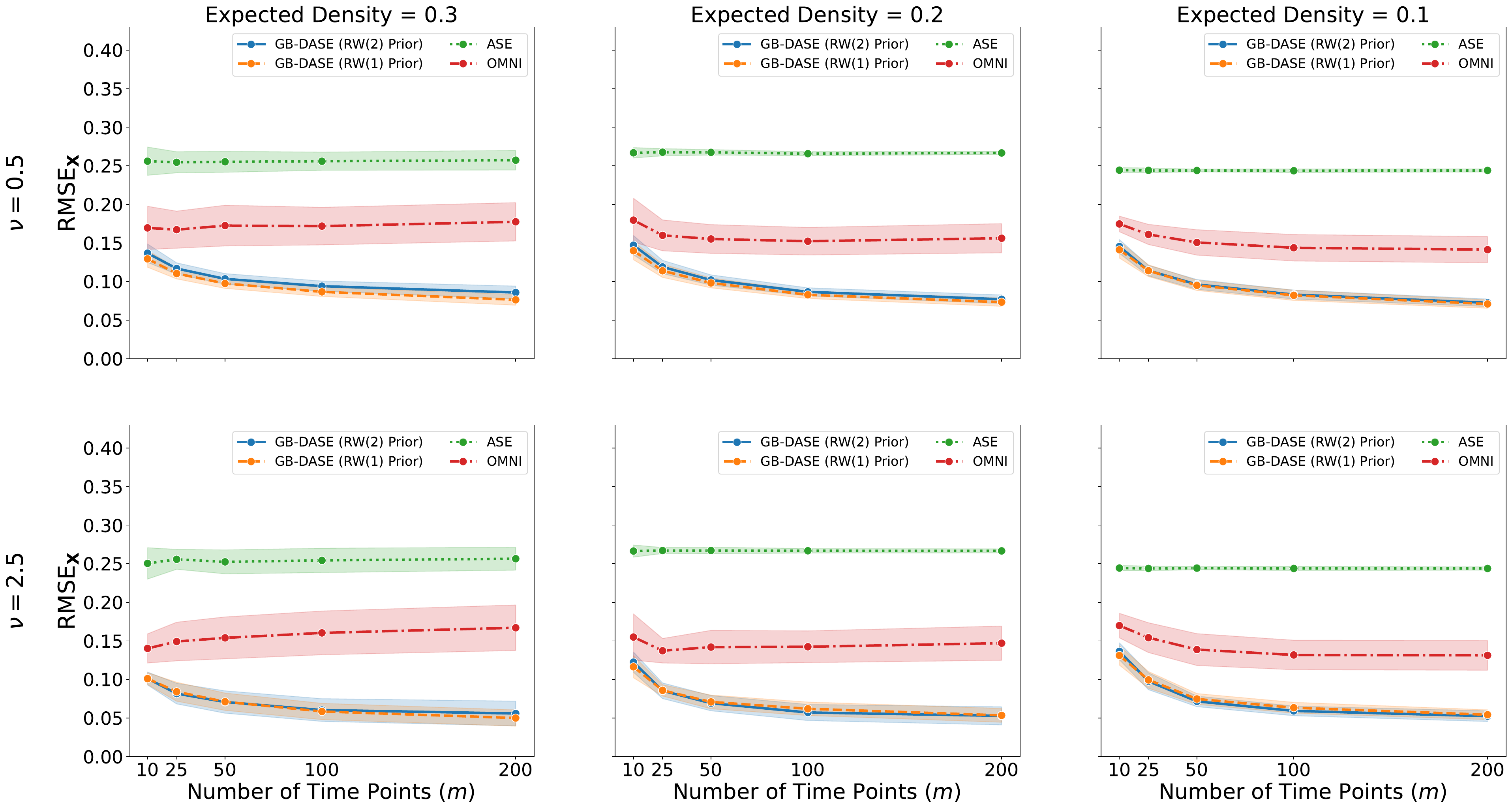}
    \caption{RMSEs for a fixed $n = 50$, varying $m$, expected edge densities, and $\nu = 0.5$ (Top) and $\nu = 2.5$ (Bottom). The curves and shaded regions indicate averages and one standard deviations over 50 independent replicates, respectively.}
    \label{fig:recovery_times}
\end{figure}

In Figure~\ref{fig:recovery_nodes}, all methods show improvement as $n$ increases, with OMNI and ASE having the slowest and fastest improvement rates, respectively. However, they never outperform the GB-DASE methods for any value of $n$ considered. Moreover, ASE gets substantially worse as the expected density decreases. In contrast, OMNI improves as the expected density shrinks. The GB-DASE methods become slightly worse as the expected density decreases, but always outperform the competitors. Overall, OMNI and GB-DASE are less sensitive to the expected density since they can compensate for the lower signal-to-noise ratio by pooling information across time points. OMNI and GB-DASE slightly improve as the latent trajectories become smoother, that is, as $\nu$ increases. The ASE's performance is relatively insensitive to $\nu$ since it does not make any assumptions about the latent trajectories' smoothness. Regardless, GB-DASE outperforms ASE for all values of $\nu$ considered. Lastly, when the signal-to-noise ratio is high enough, GB-DASE with an RW(2) prior slightly improves upon GB-DASE with an RW(1) prior for $\nu = 2.5$ and vice versa for $\nu = 0.5$, which is expected since the RW(2) prior produces smoother estimates than the RW(1) prior, as elaborated on in Section~\ref{subsec:prior}.

In Figure~\ref{fig:recovery_times}, the GB-DASE estimators improve as $m$ increases. In contrast, the ASE does not improve with increasing $m$, and OMNI's performance slightly improves before either leveling off or degrading. This phenomenon occurs because OMNI's global information sharing causes it to progressively over-smooth local changes in the latent trajectories as $m$ increases. The behavior of each estimator's performance as a function of $\nu$ is similar to  Figure~\ref{fig:recovery_nodes} discussed in the previous paragraph.

\section{Application to Forecasting International Conflicts}\label{sec:real_data}

We employ the proposed methodology to analyze a dynamic network of international conflicts between nations derived from the POLECAT global political event dataset~\citep{halterman2023}. After preprocessing, we have a dynamic network composed of the $n = 50$ most active nations or global organizations during the $m = 161$ weeks from January 2021 through January 2024. In this network, an edge ($y_{ij,t} = 1$) indicates that at least one material conflict, as defined by the PLOVER ontology~\citep{halterman2023b}, occurred between nation $i$ and nation $j$ during the $t$-th week. Under the PLOVER ontology, material conflicts encompass occurrences of economic sanctions, cyberattacks, military engagements, and other similar events.

The goal of this application is to compare GB-DASE with other ASE-based methods based on their goodness-of-fit to the dynamic conflict networks and their ability to forecast the probability of future conflicts occurring between nations. The forecasting task is particularly relevant since the POLECAT dataset is an updated version of the Integrated Crisis Early Warning System (ICEWS) dataset, which was originally curated by DARPA in the United States to develop a crisis forecasting tool~\citep{obrien2010}.  To do so, we use the first $m = 157$ weeks from January 2021 through December 2023 to estimate the competing methods. We held out the remaining $k = 4$ weeks in January 2024 to evaluate the performance of each method's $k$-step ahead forecasts. 

We estimated GB-DASE with RW(1) and RW(2) priors using 5,000 samples from Algorithm~\ref{algo:gibbs} after a burn-in of 5,000 iterations. Sampling took roughly 43 and 45 minutes on a laptop with an Apple M1 Pro processor for GB-DASE with RW(1) and RW(2) priors, respectively. The proposed GB-DASE methods use the $k$-step ahead RW($r$)-generalized posterior predictive distribution's expectation defined in Equation~(\ref{eq:gen_forecast}) as a $k$-step ahead forecast. We form a simple Monte Carlo estimate of this expectation using samples from the GB-DASE posterior. 

We compare the in-sample and $k$-step ahead forecasting performance of the proposed methods with four ASE-based competitors. The first two competitors are the ASE and OMNI estimators defined in Section~\ref{sec:simulation}. The other two methods are the unfolded ASE (UASE)~\citep{jones2021, gallagher2021}  and the multiple ASE (MASE)~\citep{arroyo2021}, which impose shared structure on the expected adjacency matrices. In particular, UASE assumes the networks come from a multilayer random dot product graph, that is, $\mathbb{E}_0(\bY_{t}) = \bV \mathbf{R}^{(t)} \bU^{(t) \, \top}$ for $\bU \in \Reals{n \times d}, \mathbf{R}^{(t)} \in \Reals{d \times d}$, and $\bV^{(t)} \in \Reals{n \times d}$. Similarly, MASE assumes the networks come from the COSIE model, that is, $\mathbb{E}_0(\bY_t) = \bV \mathbf{R}^{(t)} \bV^{\top}$ for $\bV \in \Reals{n \times d}$ and $\mathbf{R}^{(t)}\in \Reals{d \times d}$ so that the expectations share a common invariant subspace. Since existing ASE-based methods do not provide $k$-step ahead forecasts, we use the edge probability estimates at the last observed time point as a forecast for the edge probability at future time points, that is, we use $\hat{\mathbb{E}}(\bY_m)$ as a forecast for $\bY_{m+k}$.

%We considered the same ASE-based competitors from the simulation study in Section~\ref{sec:forecast_sim}, that is, the traditional ASE applied separately to each network, OMNI, UASE, and MASE. 

For each ASE-based method, we selected the latent space dimension using the profile likelihood approach of \citet{zhu2006}. This procedure performs dimension selection by locating an elbow in the scree plot of the ordered singular values of the matrix used to estimate the embedding, e.g., the omnibus matrix for OMNI. For the traditional ASE, this procedure estimates a separate dimension $\hat{d}_{\text{ASE}}^{(t)}$ for each conflict network, which are displayed in Figure~\ref{fig:ase_dhat} of the supplement. In contrast, this procedure selects a single dimension $\hat{d}_{\text{OMNI}} = 5$, $\hat{d}_{\text{UASE}} = 3$, and $\hat{d}_{\text{MASE}} = 3$, for OMNI, UASE, and MASE, respectively. Since GB-DASE uses the same spectral information as the ASE, we set the embedding dimension $d$ based on the estimated ASE embedding dimensions $\set{\hat{d}_{\text{ASE}}^{(t)}}_{t=1}^m$. Specifically, we chose $d = 4$ for GB-DASE since most ASE dimension estimates were less than or equal to four.

First, we compare the methods based on their in-sample goodness-of-fit to the conflict networks. Table~\ref{tab:auc_aupr} reports each method's overall AUC (area under the operator characteristic curve) and overall AUPR (area under the precision-recall curve) for classifying edges over all in-sample time points ($1 \leq t \leq 157$). As a baseline, we also report the performance of randomly guessing, which by definition achieves an AUC of 0.5 and an AUPR equal to the overall edge density of the dynamic network. All models greatly improve upon the random guessing baseline. Moreover, both GB-DASE methods outperform the ASE-based methods, with GB-DASE with an RW(1) prior performing best. The performance difference is especially large between the AUPRs. In practice, the AUPR is often the preferred metric because the AUC will be too optimistic when there is a large imbalance between the number of edges and non-edges~\citep{saito2015, gwee2024}, which is the case for the POLECAT conflict network that has an edge density of 7\%. For a more nuanced in-sample goodness-of-fit comparison, we quantify each method's ability to reconstruct the observed dynamic network's degree distribution using the graphical goodness-of-fit test proposed by \citet{hunter2008} in Appendix~\ref{sec:gof} of the supplement.

    \begin{table}[tb]
        \small
        \centering
    \begin{tabular}{@{}lccccccc@{}} \toprule
       Metric  & Random & GB-DASE & GB-DASE & ASE & OMNI & UASE & MASE\\
        & & [RW(1)] & [RW(2)] & & &  & \\
    \hline
        %AUC & {\bf 0.88} & 0.87 & 0.71 & 0.76 & 0.80 & 0.82 \\
        %AUPR & {\bf 0.60} & 0.50 & 0.25 & 0.36 & 0.38 & 0.39\\
        AUC & 0.5 & {\bf 0.88} & 0.86 & 0.71 & 0.77 & 0.82 & 0.80 \\
        AUPR & 0.07 & {\bf 0.59} & 0.50 & 0.25 & 0.36 & 0.39 & 0.38\\
    \bottomrule
    \end{tabular}
        \caption{AUCs and AUPRs for detecting in-sample conflicts in the conflict networks. Random is the performance of randomly guessing. Bolded entries indicate the best-performing method.}
        \label{tab:auc_aupr}
    \end{table}
    
    \begin{table}[tb]
        \small
        \centering
    \begin{tabular}{@{}cccccccc@{}} \toprule
        $k$-weeks ahead  & Random & GB-DASE & GB-DASE  & ASE & OMNI & UASE & MASE \\
         & & [RW(1)] & [RW(2)] & & &  & \\
    \hline
        %0  & 0.55 & {\bf 0.57} & 0.34 & 0.45 & 0.39 & 0.39\\
        %\hline
        %1  & 0.49 & {\bf 0.53} & 0.34 & 0.40 & 0.42 & 0.39\\
        %2  & 0.48 & {\bf 0.50} & 0.36 & 0.44 & 0.44 & 0.43 \\
        %3  & 0.47 & {\bf 0.51} & 0.30 & 0.39 & 0.32 & 0.31 \\
        %4  & 0.52 & {\bf 0.56} & 0.29 & 0.43 & 0.38 & 0.35 \\

        0  & 0.08 & 0.54 & {\bf 0.58} & 0.34 & 0.45 & 0.39 & 0.39\\
        \hline
        1  & 0.09 & 0.48 & {\bf 0.52} & 0.34 & 0.40 & 0.39 & 0.42 \\
        2  & 0.09 & 0.47 & {\bf 0.49} & 0.36 & 0.44 & 0.43 & 0.44 \\
        3  & 0.08 & 0.48 & {\bf 0.50} & 0.30 & 0.39 & 0.31 & 0.32 \\
        4  & 0.08 & 0.51 & {\bf 0.55} & 0.29 & 0.43 & 0.35 & 0.38 \\
    \bottomrule
    \end{tabular}
    \caption{AUPRs for forecasting future conflicts in the conflict networks. Random is the performance of randomly guessing. Bolded entries indicate the best-performing method.}
    \label{tab:forecast_polecat}
    \end{table}

Next, we compare the $k$-step ahead forecasting performance of the competing methods. Table~\ref{tab:forecast_polecat} displays the AUPRs for forecasting edges in the conflict network $k$-weeks ahead of the last in-sample network. For reference, the row with $k = 0$ indicates each method's AUPR for classifying edges in the last in-sample network. In this case, the GB-DASE methods out perform the competitors for all $k$-step ahead forecasts. Furthermore, GB-DASE with an RW(2) prior performs the best, which indicates a linear forecasting function is appropriate for this dataset. To further demonstrate the performance differences, we analyze the conflict probability estimates between nations involved in the Gaza war that began near the end of the in-sample observation period on October 7th, 2023.

Figures~\ref{fig:un_is} and \ref{fig:pl_us} show the estimated edge probabilities between Israel and the United Nations and the Palestinian Territories and the United States, respectively. For visual clarity, these plots only display the estimates for GB-DASE with an RW(2) prior. Appendix~\ref{sec:add_figures} of the supplement contains the estimates for GB-DASE with an RW(1) prior. The ASE estimates produce sporadic estimates due to the individual networks' low expected density and contain large jumps that do not always align the observed time series' behavior. OMNI does better in capturing the Israel and United Nations conflicts, but tends to smooth over bursts of conflict between the Palestinian Territories and the United States. UASE and MASE both appear to over-smooth the edge probabilities in both time series, with only minor changes occurring during major bursts or lulls in pairwise conflicts. The GB-DASE estimates tend to track the time series' movements well; however, they are slower to react to the change point occurring on the week of October 7th, 2023 due to the influence of the prior's local smoothness. As expected, GB-DASE with an RW(1) prior produces rougher estimates that suffer less from over-smoothing, which partly explains its' better in-sample performance. 

Lastly, we comment on the forecasts from GB-DASE with an RW(2) prior. The point forecasts predict an upward trend in the conflict probabilities, which matches the escalation of the conflicts during that time. Moreover, the widths of the 95\% credible intervals are narrow enough to be informative and increase as the forecast window $k$ increases. The width of the interval at a given $k$ in Figure~\ref{fig:un_is} is narrower than the width in Figure~\ref{fig:pl_us}, which is reasonable since the inferred edge probabilities between Israel and the United Nations are less variable than those between the Palestinian Territories and the United States.

\begin{figure}[tbp]
\centering \includegraphics[height=0.34\textheight, width=\textwidth, keepaspectratio]{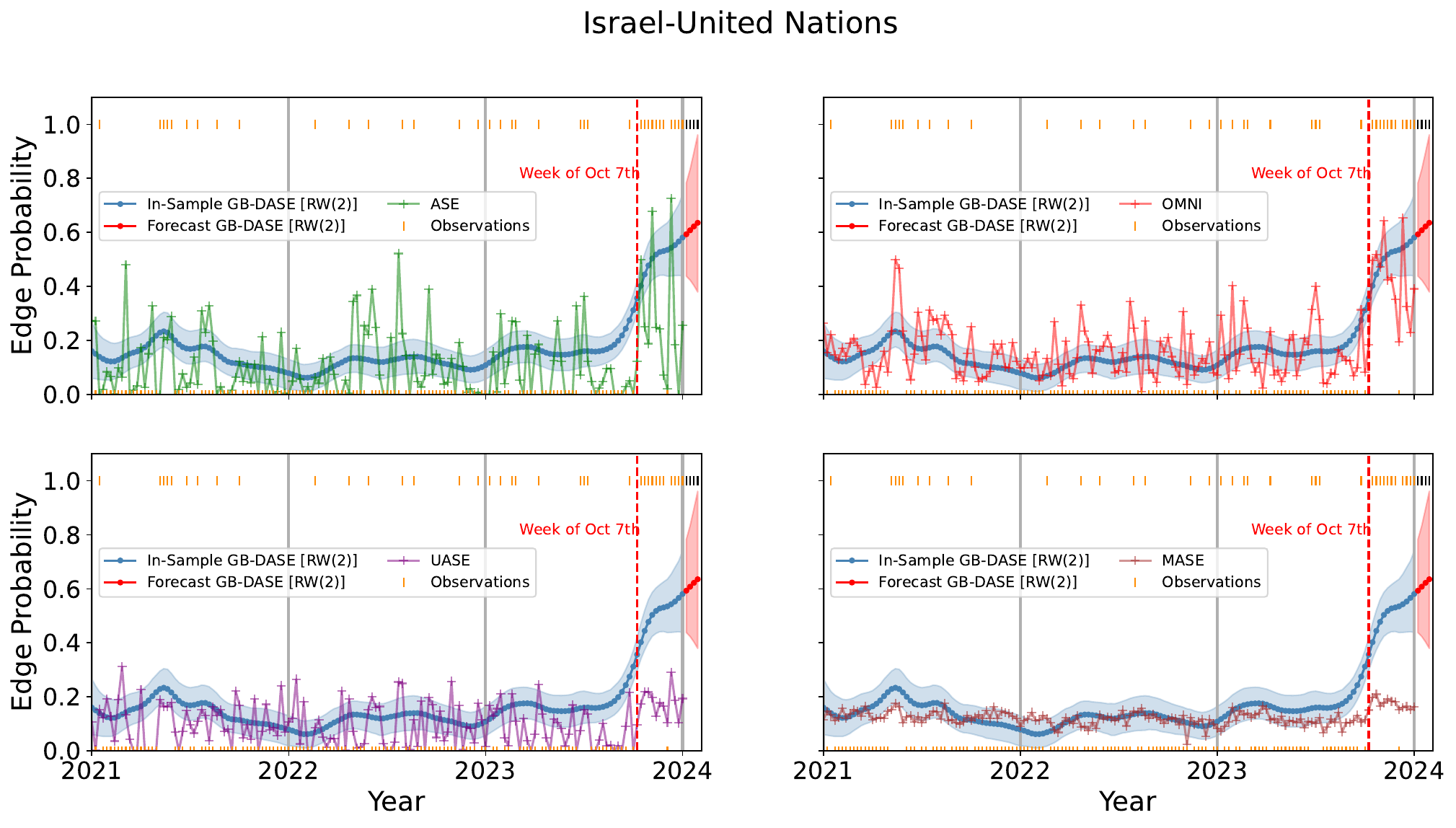}
    \caption{Estimated conflict probabilities between Israel and the United Nations. The curves denote point estimates. The blue and red shaded regions indicate 95\% pointwise credible intervals for in-sample and out-of-sample probabilities, respectively. The orange and black ticks indicate whether a conflict occurred ($y_{ij,t} = 1$) or did not occur ($y_{ij,t} = 0$) on a given week. Orange and black ticks indicate in-sample and out-of-sample weeks, respectively. The dashed horizontal line denotes the week of October 7th, 2023.}
    \label{fig:un_is}
\end{figure}

\begin{figure}[tbp]
\centering \includegraphics[height=0.34\textheight, width=\textwidth, keepaspectratio]{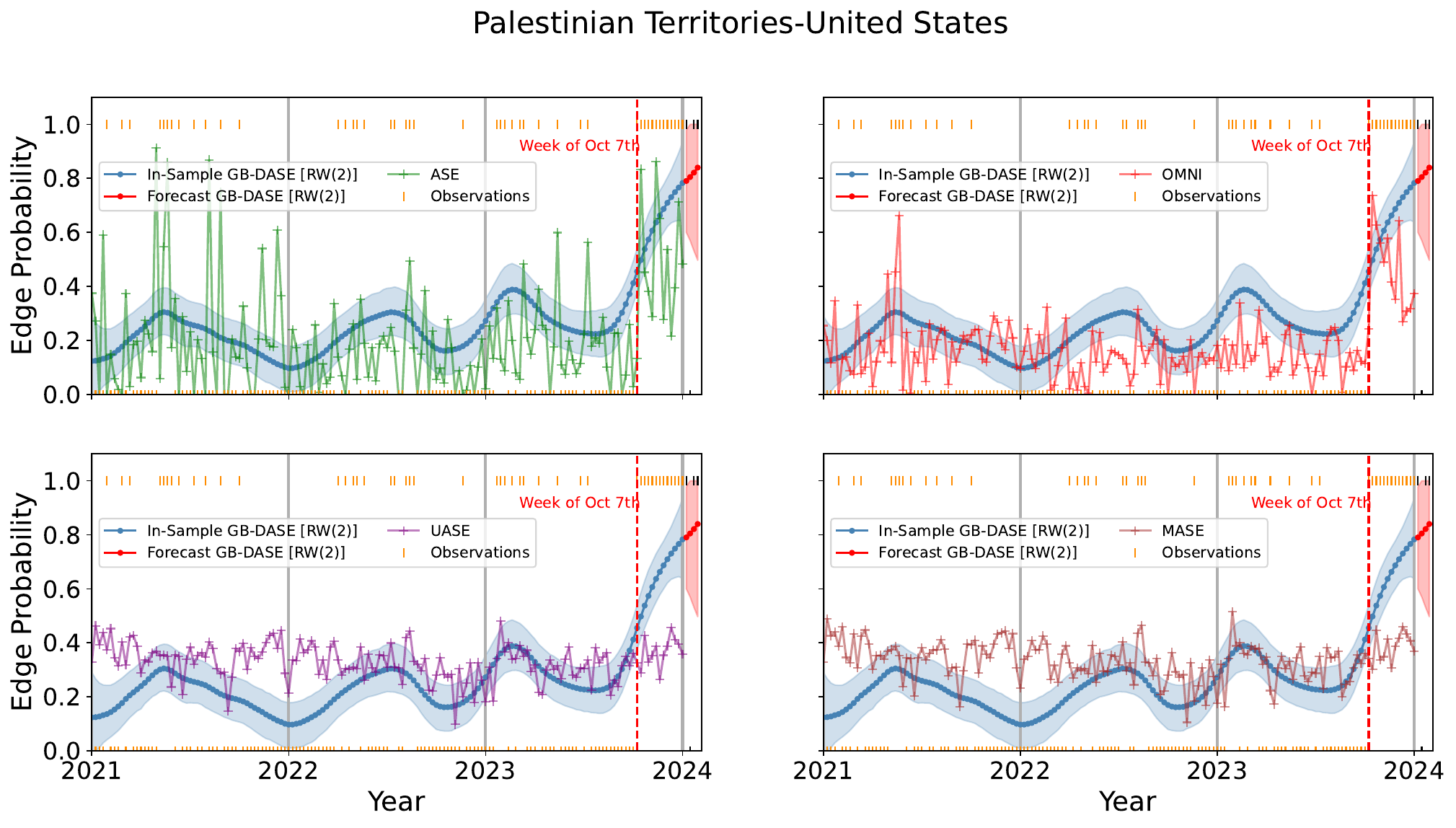}
    \caption{Estimated conflict probabilities between the Palestinian territories and the United States. The curves denote point estimates. The blue and red shaded regions indicate 95\% pointwise credible intervals for in-sample and out-of-sample probabilities, respectively. The orange and black ticks indicate whether a conflict occurred ($y_{ij,t} = 1$) or did not occur ($y_{ij,t} = 0$) on a given week. Orange and black ticks indicate in-sample and out-of-sample weeks, respectively. The dashed horizontal line denotes the week of October 7th, 2023.}
    \label{fig:pl_us}
\end{figure}

\section{Discussion}\label{sec:discussion}

This work introduced a generalized Bayesian framework for dynamic RDPGs with theoretical guarantees that provides a coherent way to update one's prior beliefs based on data and make forecasts with meaningful uncertainty quantification. Moreover, we introduced a Gibbs sampler that scales with the number of observed edges in the dynamic network, unlike existing samplers for dynamic LSMs. Lastly, we demonstrated the efficacy of the proposed approach compared to existing ASE-based methods using simulated and real networks.

There are various avenues for future research. Although we focused on smoothness-promoting random walk priors to develop informative forecasts, the proposed framework can use other temporal priors when appropriate for the statistical task. For example, the proposed framework can easily be modified to detect network change-points with meaningful uncertainty quantification by placing Markov-switching process priors on the latent positions~\citep{park2020, casarin2023}. In addition, the expected loss function in Equation~(\ref{eq:expect_loss}) used to define the generalized posterior can be extended to include the generalized RDPG~\citep{rubindelanchy2022} when a model that possesses heterophilic connectivity is appropriate. We leave these extensions to future work.

\bibliographystyle{apalike}
\bibliography{bibliography.bib}

\clearpage

\setcounter{page}{1}
\setcounter{section}{0}
\setcounter{equation}{0}
\setcounter{table}{0}
\setcounter{figure}{0}
\pagenumbering{arabic}
\renewcommand{\thesection}{\Alph{section}}
\renewcommand{\thesubsection}{\Alph{section}.\arabic{subsection}}
\renewcommand\theequation{S.\arabic{equation}}
\renewcommand\thetable{S.\arabic{table}}
\renewcommand\thefigure{S.\arabic{figure}}
\renewcommand{\thetheorem}{S.\arabic{theorem}}
\renewcommand{\thelemma}{S.\arabic{lemma}}

\begin{center}
{\Large\bf Supplementary Material for \\
    ``Generalized Bayesian Inference for Dynamic Random Dot Product Graphs''}\\[2em]
{\large Joshua Daniel Loyal\\
Department of Statistics, Florida State University}
\end{center}

\section{Additional Details on the MCMC Algorithm}\label{sec:mcmc_details}

\subsection{Derivation of the Full-Conditional Distribution of $\vecop(\bx_{i,1:m})$}

To develop the Gibbs sampler in Algorithm~\ref{algo:gibbs} of the main text, we leverage the fact that the Gibbs posterior can be viewed as a Bayesian update of a Gaussian likelihood, which allows us to recast our model as a Gaussian state-space model. Specifically, we obtain a linear observation equation when block-sampling each node $i$'s latent trajectory $\bx_{i,1:m} = (\bx_{i1}, \dots, \bx_{im}) \in \Reals{d \times m}$ conditioned on the remaining latent trajectories $\set{\bx_{j,1:m} : 1 \leq j \neq i \leq n}$ and model parameters. A few existing discrete-time LSMs can be recast as Gaussian state-space models for a transformed response~\citepSup{loyal2023, zhao2024}. These models apply the Kalman filter and smoother to sample from the latent trajectories' full-conditional distributions in $O(m)$ time by accounting for the sparsity in the full-conditional's precision matrix. However, this approach does not take advantage of the additional sparsity in the observed network, so it scales quadratically in $n$. As such, we develop a new approach that takes advantage of both the sparsity in the precision matrix and the observed network.

To proceed, we derive each latent trajectory's full-conditional distribution. Based on the Gaussian pseudolikelihood in Equation~(\ref{eq:gauss_gibbs_post}) of the main text, for each node $i = 1, \dots ,n$, given the other latent trajectories $\set{\bx_{j,1:m} : 1 \leq j \neq i \leq n}$, $\sigma_{1:n}$, and $\lambda$, we obtain the following linear observation equations in $\bx_{it}$:
\begin{equation*}
    \by_{it} = \bX_{t}^{(-i)}\,  \bx_{it} + \bm{\varepsilon}_{it}, \qquad \bm{\varepsilon}_{it} \indsim N(\mathbf{0}_{n-1},\, \lambda^{-1} \, \mathbf{I}_{n-1}) \qquad (t = 1, \dots, m),
\end{equation*}
where $\bX_{t}^{(-i)} \in \Reals{(n-1) \times d}$ is the latent position matrix at time $t$ with the $i$-th row removed and similarly $\by_{it} \in \Reals{n-1}$ is the $i$-th column of the observed adjacency matrix at time $t$ with the $i$-th row removed. Note that since the diagonal elements of the adjacency matrices are non-random,  we absorb their contributions into the prior. The previous equation can be compactly expressed as a single linear regression equation:
\[
    \by_i = \diag\{\bX_1^{(-i)}, \dots, \bX_m^{(-i)}\} \  \text{vec}(\bx_{i,1:m}) + \bm{\varepsilon}_i,
\]
where $\by_i = (\by_{i1}^{\top}, \dots, \by_{im}^{\top})^{\top} \in \Reals{m(n-1)}$ and $\bm{\varepsilon}_i = (\bm{\varepsilon}_{i1}^{\top}, \dots, \bm{\varepsilon}_{im}^{\top})^{\top} \in \Reals{m(n-1)}$. 
Recalling the random walk priors in Equation~(\ref{eq:x_prior}) of the main text, standard manipulations for Bayesian linear regression demonstrate that the latent trajectories have the following multivariate normal full-conditional distribution for each $i = 1, \dots, n$:
\begin{equation}\label{eq:full_con}
    p(\text{vec}(\bx_{i,1:m}) \mid \cdot) =  N\{\text{vec}(\bx_{i,1:m}) \mid \text{vec}(\bmu_{i,1:m}), \mathbf{P}_i^{-1}\},
\end{equation}
with precision matrix
\begin{align}\label{eq:con_prec}
    \mathbf{P}_i &= \{(1/\sigma_i^2) \bD_r^{\top}\bD_r  + (1/\sigma_0^2) \sum_{s=1}^r \mathbf{e}_s \mathbf{e}_s^{\top}\} \otimes \mathbf{I}_d + (\lambda/2) \mathbf{I}_{md} \nonumber \\
    &\qquad\qquad+ \lambda \, \diag\{\bX_1^{(-i)}, \dots, \bX_m^{(-i)}\}^{\top} \diag\{\bX_1^{(-i)}, \dots, \bX_m^{(-i)}\} \nonumber \\
    &= \{(1/\sigma_i^2) \bD_r^{\top}\bD_r  + (1/\sigma_0^2) \sum_{s=1}^r \mathbf{e}_s \mathbf{e}_s^{\top}\} \otimes \mathbf{I}_d + (\lambda/2) \mathbf{I}_{md} \nonumber \\
    &\qquad\qquad + \lambda \sum_{j\neq i} \diag(\bx_{j1}\bx_{j1}^{\top}, \dots, \bx_{jm} \bx_{jm}^{\top}), 
\end{align}
and a mean vector that satisfies the following linear system of equations
\begin{equation}\label{eq:con_mean}
    \mathbf{P}_i \text{vec}(\bmu_{i,1:m}) = \lambda \diag\{\bX_1^{(-i)}, \dots, \bX_m^{(-i)}\}^{\top}\by_i = \lambda  \begin{pmatrix}
        \sum_{j \in \mathcal{N}_{i1}} \bx_{j1} y_{ij,1} \\
        \vdots \\
        \sum_{j \in \mathcal{N}_{im}} \bx_{jm} y_{ij,m},
\end{pmatrix},
\end{equation}
where $\mathcal{N}_{it} = \set{j \, : \, y_{ij,t} \neq 0}$ is the neighborhood of node $i$ at time $t$. This full-conditional distribution is sampled for each $i = 1, \dots, n$ in Step 1 of Algorithm~\ref{algo:gibbs} of the main text. 

%Crucially, $\bP_i \in \Reals{md \times md}$ is a banded matrix with bandwidth $rd$, which follows from fact that $\bP_i$ is a sum of four banded matrices with bandwidths $rd$, zero, zero, and $d - 1$, respectively.

\subsection{Derivation of the Full-Conditionals for $\sigma_{1:n}$ and $\lambda$}

Here, we derive the full-conditional distributions for the node-wise variances $\sigma_{1:n}$ and the learning rate $\lambda$. To sample the node-wise variances, we use a parameter expansion strategy based on the scale mixture representation of the half-Cauchy distribution~\citepSup{wand2011}, that is, if 
\[
    \sigma_i^2 \mid \nu_i \indsim \text{Inverse-Gamma}(1/2, 1/\nu_i), \quad \nu_i \iidsim \text{Inverse-Gamma}(1/2, 1) \quad (i = 1, \dots, n),
\]
then marginally $\sigma_i \iidsim C^{+}(0, 1)$ for $i = 1, \dots, n$. Including this additional hierarchy in our model results in the conjugate full-conditional updates in Step (2.a) and Step (2.b) of Algorithm~\ref{algo:gibbs} through standard manipulations. 

Lastly, the $\text{Gamma}(a_{\lambda}, b_{\lambda})$ prior on $\lambda$ is conjugate, so it produces a gamma full-conditional distribution with shape $a_{\lambda} + n(n+1)m/4$ and rate $b_{\lambda} + (1/4) \sum_{t=1}^m \frobnorm{\bY_m - \bX_m\bX_m^{\top}}^2$. Although the rate parameter requires $O(mn^2d)$ operations to compute, it can be efficiently vectorized using specialized BLAS routines and thus accounts for a negligible proportion of the algorithm's runtime. For example, on a laptop with an Apple M1 Pro processor, Step 3 in Algorithm~\ref{algo:gibbs} takes about 0.1 seconds, and a complete iteration of Steps 1--3 of Algorithm~\ref{algo:gibbs} takes roughly one second for a network with an average edge density of $0.1$, $n = 400$ nodes, $m = 100$ time points, and $d = 2$ latent dimensions.  However, for very large networks, Step 3 can be expensive. In this case, we recommend employing an empirical Bayesian approach and fixing $\lambda$ to one over the sample variance of the observed edge-variables. We found this approach performed only slightly worse than sampling $\lambda$.

\subsection{Initialization and Post-Processing}\label{sec:init}

Good initial values for the model parameters can greatly improve the mixing time of  MCMC samplers for high-dimensional posteriors. Recall that the model parameters are $\bX_{1:m}, \sigma_{1:n}^2, \nu_{1:n}$, and $\lambda$. We initialized the latent positions to the ASE applied separately to each adjacency matrix, that is, we set $\bX_{1:m}^{(0)} = \hat{\bX}_{1:m}^{(\mathrm{ASE})}$ defined in Equation~(\ref{eq:ase}) of the main text. To ensure the initial latent trajectory is smooth, we aligned the neighboring latent position matrices using sequential Procrustes rotations~\citepSup{hoff2002}. Specifically, moving sequentially forward in time, starting at $t = 2$, we projected $\bX_t^{(0)}$ to the locations that are closest to its previous location $\bX_{t-1}^{(0)}$ through a Procrustes rotation. We initialized the node-wise transition variances $\sigma_i^{2 \, (0)} = \{(m-1)d\}^{-1} \sum_{t=2}^m \norm{\bx_{it}^{(0)} - \bx_{i(t-1)}^{(0)}}^2_2$ for $i = 1, \dots, n$. We set $\nu_i^{(0)} = 1$ for $i = 1, \dots, n$. Lastly, we set $\lambda^{(0)}$ to one over the sample variance of the observed edge-variables based on its role as a precision parameter, as discussed in Section~\ref{subsec:lambda} of the main text.

As discussed in Section~\ref{sec:rdpg} of the main text, the latent positions are only identified up to unknown orthogonal transformations. This non-identifiability remains in the generalized posterior and can lead to inflated variance estimates without proper post-processing. As such, we post-process each posterior draw $\bX_{1:m}^{(s)}$ by matching it to some reference set of latent positions $\bX_{1:m}^{\dagger}$ using Procrustes matching, which is the standard solution in the Bayesian latent space literature~\citepSup{hoff2002}. For convenience, we set the reference latent positions $\bX_{1:m}^{\dagger}$ to the last posterior draw $\bX_{1:m}^{(S)}$ after sequentially aligning its latent positions, that is, moving sequentially forward in time starting at $t = 2$, we rotated $\bX_{t}^{(S)}$ to the location closest to its previous location $\bX_{t-1}^{(S)}$ through a Procrustes rotation. Then, we employed Procrustes rotations to individually match each $\bX_t^{(s)}$ to $\bX_t^{\dagger}$ for all $t = 1, \dots, m$ and $s = 1, \dots, S$.

\section{Proof of Theorem~1}\label{sec:proofs}

\subsection{Additional Notation}

First, we establish some additional notation used throughout this section. We use $\frobket{\cdot, \cdot}$ to denote the Frobenius inner product. For positive sequences $\set{a_n}$ and $\set{b_n}$, we write $a_n = O(b_n)$ if there exists a constant $C > 0$ such that $a_n \leq C b_n$ for $n$ large enough. We use $a_n \lesssim b_n$ (or $b_n \gtrsim a_n$) to indicate that $a_n = O(b_n)$. We say $a_n = o(b_n)$ if $\lim_{n\rightarrow \infty} a_n/b_n = 0$. 
\subsection{Preliminary Lemmas}

    This section contains preliminary results used in the proof of the main theorem. The following lemma lower-bounds the small-ball probabilities of the first-order random walk prior, whose proof is found in \citetSup{zhao2024}.

\begin{lemma}[Proof of Theorem 3 on page 31 and Lemma 13 of \citetSup{zhao2024}]\label{lemma:rw_concentration}
    If the components of $\bx = (x_1, \dots, x_m)^{\top}$ follow a first-order Gaussian random walk with transition variance $\sigma^2$, that is,
    \[
        x_t \mid x_{t-1} \sim N(x_{t-1}, \sigma^2) \quad (t = 2, \dots, m),
    \]
    then for any $\delta > 0$, 
    \[
        \mathbb{P}\left(\max_{2 \leq t \leq m} \abs{x_t} \leq \delta \right) \geq e^{-Cm\sigma^2/\delta^2},
    \]
    where $C = 2 (1 + 3 e \sqrt{2 \pi})^2$.
\end{lemma}

    The next lemma is the key concentration inequality used to control the size of the generalized posterior distribution, whose proof can be found in \citetSup{xie2020}.

\begin{lemma}[Lemma E.1 of \citetSup{xie2020}]\label{lemma:concen}
    Let $\bE \in \Reals{n \times n}$ be a symmetric random matrix with $\mathbb{E}_0(\bE) = \mathbf{0}_{n \times n}$. Assume that $\bE$ is sub-Gaussian, that is, there exists some constant $\tau > 0$, such that for all $\bA \in \Reals{n \times n}$ with $\frobnorm{\bA}^2 = 1$, for all $t > 0$, $\mathbb{P}_0(\abs{\tr(\bA^{\top}\bE)} > t) \leq 2e^{-\tau t^2}$. For any $\bX, \bX_0 \in \Reals{n \times d}$ and for any $t > 0$,
    \[
        \mathbb{P}_0\left(\sup_{\bX \in \Reals{n \times d}} \, \abs*{\frobket*{\bE, \frac{\bX\bX^{\top} - \bX_0 \bX_0^{\top}}{\frobnorm{\bX\bX^{\top} - \bX_0 \bX_0^{\top}}}}} > nt \right) \leq 6 \exp\left(3nd - \frac{\tau n^2 t^2}{4}\right).
    \]
\end{lemma}

Lastly, we state the following lemma that relates two common metrics for comparing matrices.
\begin{lemma}[Lemma 5.4 in \citetSup{tu2016}]\label{lemma:tu2016}
    For any $\bX, \bX_0 \in \Reals{n \times d}$, we have
    \[
        \min_{\bW \in \mathcal{O}_d} \frobnorm{\bX - \bX_0 \bW}^2 \leq  \frac{1}{2(\sqrt{2} - 1) \sigma_d^2(\bX_0)} \frobnorm{\bX \bX^{\top} - \bX_0\bX_0^{\top}}^2,
    \]
    where $\sigma_d(\bX_0)$ is the smallest singular value of $\bX_0$.
\end{lemma}

\subsection{Prior Support}

    The next results establish the support of the first-order random walk prior distributions over small balls around the true latent positions. These results are used to control the size of the denominator $D_{nm}$ of the generalized posterior distribution defined in Section~\ref{sec:theory} of the main text. Lemma~\ref{lemma:prior_support_sigma} demonstrates that the RW(1) prior conditioned on specific transition variances, that is, $\Pi_{\text{RW}(1)}(\cdot \mid \sigma_{1} = \sigma^*, \dots, \sigma_n = \sigma^{*}, \sigma_0)$, has appropriate support. Lemma~\ref{lemma:prior_support} demonstrates that the half-Cauchy prior for the transition variances places enough mass around these specific variances so that the result extends to the marginal prior $\Pi_{\text{RW}(1)}(\cdot \mid \sigma_0)$.

\begin{lemma}\label{lemma:prior_support_sigma}
    Define the set  $\mathcal{B}_{nm}(\gamma) = \set{\bX_{1:m} \, : \, \max_{1 \leq t \leq m} \frobnorm{\bX_t \bX_t^{\top} - \bX_{0t} \bX_{0t}^{\top}} \leq C_{\bx} \sqrt{d} n \gamma}$ for $\set{\bX_{0t}}_{t=1}^m\in \pawd(\bL)$ and $\gamma > 0$. Let 
    \[
        \epsilon_{nm} = \frac{\norm{\bL}_{\infty}^{1/3}}{m^{1/3}n^{2/3}} + \sqrt{\frac{\log nm}{nm}}
    \]
    and $\sigma^{*\, 2} = c_0 \epsilon_{nm} d \tilde{L} / (nm)$ for any positive constant $c_0 > 0$ and $\tilde{L} = \max\{\norm{\bL}_{\infty}, \log^{3/2}(nm) \sqrt{n/m}\}$. Under the first-order random walk prior with initial variance $\sigma_0$ and transition variances $\sigma_1, \dots, \sigma_n$,
    \[
        \Pi_{\mathrm{RW}(1)}\{\mathcal{B}_{nm}(\epsilon_{nm}) \mid \sigma_1 = \sigma^*, \dots, \sigma_n = \sigma^*, \sigma_0\} \geq e^{-C d n^2 m \epsilon_{nm}^2}
    \]
    for a constant $C$ that only depends on $\sigma_0$.
\end{lemma}
\begin{proof}
    Observe that for all $1 \leq t \leq m$,
    \begin{align*}
        \frobnorm{\bX_t \bX_t^{\top} - \bX_{0t} \bX_{0t}^{\top}} &\leq \frobnorm{(\bX_t - \bX_{0t})(\bX_t - \bX_{0t})^{\top}} + \frobnorm{\bX_{0t}(\bX_t - \bX_{0t})^{\top}} + \frobnorm{(\bX_t - \bX_{0t})\bX_{0t}^{\top}} \\
        &\leq \frobnorm{\bX_t - \bX_{0t}}^2 + 2 C_{\bx} \sqrt{n}  \frobnorm{\bX_t - \bX_{0t}} \\
        &= (2 C_{\bx} \sqrt{n} + \frobnorm{\bX_t - \bX_{0t}}) \frobnorm{\bX_t - \bX_{0t}},
    \end{align*}
    where we used the fact that $\max_{1 \leq t \leq m}\norm{\bx_{0it}}_2 \leq C_{\bx}$ under the PAWD($\bL$) condition. In addition,
    \[
        \frobnorm{\bX_t - \bX_{0t}} \leq \frobnorm{\bX_1 - \bX_{01}} + \frobnorm{(\bX_t - \bX_1) - (\bX_{0t} - \bX_{01})} = \frobnorm{\bX_1 - \bX_{01}} + \frobnorm{\tilde{\bX}_t - \tilde{\bX}_{0t}},
    \]
    where $\tilde{\bX}_t = \bX_t - \bX_1$ and $\tilde{\bX}_{0t} = \bX_{0t} - \bX_{01}$. 
    It follows that
    \begin{align*}
        \bigg\{\max_{1 \leq i \leq n} \norm{\bx_{i1} - \bx_{0i1}}_2 \leq \frac{\sqrt{d} \epsilon_{nm}}{6}, &\max_{1 \leq i \leq n, 1 \leq h \leq d} \norm{\tilde{\bx}_{i,h} - \tilde{\bx}_{0i,h}}_{\infty} \leq \frac{\epsilon_{nm}}{6}\bigg\} \\
        &\qquad\subset \left\{\max_{1 \leq t \leq m} \frobnorm{\bX_t - \bX_{0t}} \leq \frac{\sqrt{dn} \epsilon_{nm}}{3} \right\} \\
        &\qquad\subset \left\{\max_{1 \leq t \leq m}\frobnorm{\bX_t \bX_t^{\top} - \bX_{0t} \bX_{0t}^{\top}} \leq C_{\bx} \sqrt{d} n \epsilon_{nm} \right\}\\
        &\qquad=\mathcal{B}_{nm}(\epsilon_{nm}),
    \end{align*}
    where $\bx_{ih} = (x_{ih,1}, \dots, x_{ih,m})^{\top}$ and $\bx_{0ih} = (x_{0ih,1}, \dots x_{0ih,m})^{\top}$.
    Under the independence assumptions of the first-order random walk priors, we have
    \begin{align}\label{eq:support_lb}
        \Pi_{\mathrm{RW}(1)}\{\mathcal{B}_{nm}(\epsilon_{nm}) &\mid  \sigma_1 = \sigma^*, \dots, \sigma_n = \sigma^*, \sigma_0\} \geq \nonumber \\
        &\Pi\left(\max_{1 \leq i \leq n} \norm{\bx_{i1} - \bx_{0i1}}_2 \leq \frac{\sqrt{d} \epsilon_{nm}}{6} \, \bigg| \, \sigma_0\right) \times \nonumber\\
        &\quad\Pi\left(\max_{1 \leq i \leq n, 1 \leq h \leq d} \norm{\tilde{\bx}_{i,h} - \tilde{\bx}_{0i,h}}_{\infty} \leq \frac{\epsilon_{nm}}{6}\, \bigg| \, \sigma_1 = \sigma^*, \dots, \sigma_n = \sigma^* \right).
    \end{align} 
    
    We begin with the first term. Recall that we place mean-zero Gaussian priors on the initial latent positions. The concentration of a Gaussian distribution can be lower bounded by Anderson's lemma:
    \begin{align*}
        \Pi\left(\max_{1 \leq i \leq n} \norm{\bx_{i1} - \bx_{0i1}}_2 \leq \frac{\sqrt{d}\epsilon_{nm}}{6} \, \bigg| \, \sigma_0\right)  &\geq \exp\left(-\frac{1}{2\sigma_0^2}\frobnorm{\bX_0}^2\right)\prod_{i=1}^n \Pi\left(\norm{\bx_{i1}}_2 \leq \frac{\sqrt{d}\epsilon_{nm}}{6} \, \bigg|\, \sigma_0\right) \\
        &\geq \exp\left(-\frac{1}{2\sigma_0^2}\frobnorm{\bX_0}^2\right)\prod_{i=1}^n \prod_{h=1}^d \Pi\left(x_{i1,h}^2 \leq \frac{\epsilon_{nm}^2}{36} \, \bigg|\, \sigma_0\right) \\
        &\geq \exp\left(-\frac{n C_{\bx}^2}{2\sigma_0^2}\right) \left(\frac{1}{\sqrt{2\pi} \sigma_0}\right)^{nd} \left(\frac{\epsilon_{nm}}{6}\right)^{nd} \\
        &\geq \exp\left\{-\left(\frac{C_{\bx}^2}{2\sigma_0^2} + d\right)n -nd \abs*{\log \frac{\epsilon_{nm}}{6 \sigma_0}}\right\} \\
        &\geq \exp\left\{-\left(\frac{C_{\bx}}{2\sigma_0^2} + d + d \abs{\log 6 \sigma_0} \right)n - nd \log \frac{1}{\epsilon_{nm}}\right\}.
    \end{align*}

    For the second term in Equation~(\ref{eq:support_lb}), we again apply Anderson's lemma:
    \begin{align*}
        \Pi\bigg(&\max_{1 \leq i \leq n, 1 \leq h \leq d} \norm{\tilde{\bx}_{i,h} - \tilde{\bx}_{0i,h}}_{\infty} \leq \frac{\epsilon_{nm}}{6} \, \bigg| \, \sigma_1 = \sigma^*, \dots, \sigma_n = \sigma^* \bigg) \\
        &\geq \prod_{i=1}^n \prod_{h=1}^d \exp\left(-\frac{1}{2\sigma^2_i} \norm*{\bD_1 \tilde{\bx}_{0i,h}}_2^2 \right) \Pi\left(\norm{\tilde{\bx}_{i,h}}_{\infty} \leq \frac{\epsilon_{nm}}{6}\, | \, \sigma_1 = \sigma^*, \dots, \sigma_n = \sigma^*\right) \\
        &\geq \exp\left(-\frac{1}{2}\sum_{i=1}^n \sum_{t=1}^m \frac{\norm{\bx_{0it} - \bx_{0i(t-1)}}_2^2}{\sigma^{* \,2}}\right) \prod_{i=1}^n \prod_{h=1}^d\Pi\left(\norm{\tilde{\bx}_{i,h}}_{\infty} \leq \frac{\epsilon_{nm}}{6} \, \bigg|\, \sigma_1 = \sigma^*, \dots, \sigma_n = \sigma^* \right)\\
        &\geq \exp\left(-\frac{d^2}{2n^2m} \sum_{i=1}^n \frac{L_i^2}{\sigma^{* \, 2}} - 36 C m \sum_{i=1}^n \frac{\sigma^{* \, 2}}{\epsilon^2_{nm}} \right),
    \end{align*}
    where the last inequality follows from the $\textrm{PAWD}(\bL)$ condition and Lemma~\ref{lemma:rw_concentration}. Setting $\sigma^{*\, 2} =  c_0 \epsilon_{nm}  d \tilde{L} / (nm)$, we observe that
    \[
        \Pi\bigg(\max_{1 \leq i \leq n, 1 \leq h \leq d} \norm{\tilde{\bx}_{i,h} - \tilde{\bx}_{0i,h}}_{\infty} \leq \frac{\epsilon_{nm}}{6} \,\bigg|\, \sigma_1 = \sigma^*, \dots, \sigma_n = \sigma^* \bigg) \geq \exp\left\{-C_1 d \frac{\tilde{L}}{\epsilon_{nm}}\right\},
    \]
    where $C_1 = -\{1/(2c_0) + 36 C c_0 \}$. Combining these inequalities, we have
    \begin{align*}
        \Pi_{\mathrm{RW}(1)}\{\mathcal{B}_{nm}(\epsilon_{nm}) \mid &\sigma_1 = \sigma^*, \dots, \sigma_n = \sigma^*, \sigma_0\} \\
        &\geq \exp\left\{-\left(\frac{C_{\bx}}{2\sigma_0^2} + d + d \abs{\log 6 \sigma_0} \right)n - nd \log \frac{1}{\epsilon_{nm}}- C_1 d \frac{\tilde{L}}{\epsilon_{nm}}\right\}.
    \end{align*}
    The rate $\tilde{\epsilon}_{nm} = \tilde{L}^{1/3}m^{-1/3}n^{-2/3} + \{\log(nm)/(nm)\}^{1/2}$ is obtained by finding the smallest possible $\tilde{\epsilon}_{nm}$ such that $n^2 m \tilde{\epsilon}_{nm}^2 \gtrsim \max\{\tilde{L}/\tilde{\epsilon}_{nm}, n \log(1/\tilde{\epsilon}_{nm})\}$. Furthermore, $\tilde{\epsilon}_{nm} \leq 2 \epsilon_{nm}$ and $\{(C_{\bx}/2\sigma_0^2) + d + d \abs{\log 6 \sigma_0}\}n \lesssim d n^2 m \epsilon_{nm}^2$. As such,
    \[
        \Pi_{\mathrm{RW}(1)}\{\mathcal{B}_{nm}(\epsilon_{nm}) \mid \sigma_1 = \sigma^*, \dots, \sigma_n = \sigma^*, \sigma_0\} \geq  e^{-C_2 dn^2m \epsilon_{nm}^2},
    \]
    for a constant $C_2$ that only depends on $\sigma_0$.
\end{proof}

\begin{lemma}\label{lemma:prior_support}
    Define the set  $\mathcal{B}_{nm}(\gamma) = \set{\bX_{1:m} \, : \, \max_{1 \leq t \leq m} \frobnorm{\bX_t \bX_t^{\top} - \bX_{0t} \bX_{0t}^{\top}} \leq C_{\bx} \sqrt{d} n \gamma}$ for $\set{\bX_{0t}}_{t=1}^m\in \mathrm{PAWD}(\bL)$ where $\norm{\bL}_{\infty} = O\{nm \log(nm)\}$. Let 
    \[
        \epsilon_{nm} = \frac{\norm{\bL}_{\infty}^{1/3}}{m^{1/3}n^{2/3}} + \sqrt{\frac{\log nm}{nm}}
    \]
    Under the first-order random walk prior with initial variance $\sigma_0$ and transition variances $\sigma_1, \dots, \sigma_n \iidsim C^{+}(0,1)$,
    \[
        \Pi_{\mathrm{RW}(1)}\{\mathcal{B}_{nm}(\epsilon_{nm}) \mid \sigma_0\} \geq e^{-C d n^2 m \epsilon_{nm}^2}
    \]
    for a constant $C$ that only depends on $\sigma_0$.
\end{lemma}
\begin{proof}
    As in Lemma~\ref{lemma:prior_support_sigma}, we define $\sigma^{*\, 2} = \epsilon_{nm} d \tilde{L} / (nm)$ with $\tilde{L} = \max\{\norm{\bL}_{\infty}, \log^{3/2}(nm) \sqrt{n/m}\}$. Limiting on the subset $\mathcal{N}(\sigma_{1:n}) = \cap_{i=1}^n \set{\sigma_i \, : \, \abs{\sigma_i - \sigma^{*}} \leq \sigma^{*} / 2}$, we have for some constant $C_1 > 0$ that only depends on $\sigma_0$,
    \begin{align*}
        \Pi_{\mathrm{RW}(1)}\{\mathcal{B}_{nm}(\epsilon_{nm}) \mid \sigma_0\} &\geq \int_{\mathcal{N}(\sigma_{1:n})} \Pi_{\mathrm{RW}(1)}\{\mathcal{B}_{nm}(\epsilon_{nm}) \mid \sigma_{1}, \dots, \sigma_n, \sigma_0\} \Pi(\rmd \sigma_{1:n}) \\
        &\geq \exp\left\{-C_1 d n^2 m \epsilon_{nm}^2\right\} \prod_{i=1}^n \Pi(\abs{\sigma_i - \sigma^{*}} \leq \sigma^{*}/2) \\
        &= \exp\left\{- C_1 d n^2 m \epsilon_{nm}^2\right\}  \Pi(\abs{\sigma_1 - \sigma^{*}} \leq \sigma^{*}/2)^n,
    \end{align*}
    where we used Lemma~\ref{lemma:prior_support_sigma} in the first line and the independence of the priors in the second line. Recall that $\sigma_1^2$ has a $C^{+}(0,1)$ prior, so that
    \begin{align*}
        \Pi(\abs{\sigma_1 - \sigma^{*}} \leq \sigma^{*}/2) &= \frac{2}{\pi} \int_{\sigma^{*}/2}^{3\sigma^{*}/2} \frac{1}{1 + \sigma_1^2} \ \rmd \sigma_1 \\
        &\geq \frac{\sigma^*}{2} \frac{1}{1 + 9 \sigma^{2 *}/4}.
    \end{align*}
    As such,
    \begin{align*}
        -n \log\{\Pi(\abs{\sigma_1 - \sigma^{*}} \leq \sigma^*/2)\} &\leq  n \log(2) - n \log \sigma^*  + n \log\{1 + 9\sigma^{2 *}/4\} \\
        &\leq n \log(2) -\frac{1}{2}n \log \sigma^{2 *} + \frac{9}{4} n \sigma^{2 *},
    \end{align*}
    where we used $\log(1 + x) \leq x$ for $x > -1$. Since $\epsilon_{nm} = o(1)$ and $\norm{\bL}_{\infty} = O\{nm \log(nm)\}$, we have
    \[
        n \sigma^{2 *} \leq \epsilon_{nm} d \tilde{L} /m \lesssim d \tilde{L}/m \lesssim dn \log(nm) \lesssim d n^2 m \epsilon_{nm}^2.
    \]
    Furthermore, 
    \[
        -n\log\sigma^{2 *} \lesssim n \log(1/\epsilon_{nm}) + n\log(m) - n\log(\tilde{L}) \lesssim n \log(nm) \lesssim dn^2 m \epsilon_{nm}^2,
    \]
    where we used the facts that $\tilde{L} \geq \log^{3/2}(nm) \{n/m\}^{1/2}$, $\epsilon_{nm} \geq \sqrt{\log(nm) / nm}$, and $mn >  1$. Combining these results, we have
    \[
        \Pi_{\mathrm{RW}(1)}\{\mathcal{B}_{nm}(\epsilon_{nm}) \mid \sigma_0\} \geq \exp\left\{-C_2 d n^2 m \epsilon_{nm}^2\right\},
    \]
    for a constant $C_2$ that only depends on $\sigma_0$.
\end{proof}

\subsection{Proof of the Main Result}

\begin{proof}[Proof of Theorem~\ref{thm:post_con}]
    To begin, for $1 \leq t \leq m$, we have the equality 
    \[
        -\frac{1}{2} \frobnorm{\bY_t - \bX_t \bX_t^{\top}}^2 + \frac{1}{2} \frobnorm{\bY_t - \bX_{0t}\bX_{0t}^{\top}}^2 = - \frac{1}{2} \frobnorm{\bX_t \bX_t^{\top} - \bX_{0t} \bX_{0t}^{\top}}^2 + \frobket{\bE_t, \bX_t \bX_t^{\top} - \bX_{0t} \bX_{0t}^{\top}},
    \]
    where $\bE_t = \bY_t - \mathbb{E}_0 \bY_t$. For any $\alpha > 0$, denote the event
    \[
        \mathcal{E}_{nm}(\alpha) = \bigcap_{t=1}^m \left\{\sup_{\bX_t \in \Reals{n \times d}} \abs*{\frobket*{\bE_t, \frac{\bX_t\bX_t^{\top} - \bX_{0t} \bX_{0t}^{\top}}{\frobnorm{\bX_t\bX_t^{\top} - \bX_{0t} \bX_{0t}^{\top}}}}} \leq \alpha \sqrt{d} n \epsilon_{nm} \right\}.
    \]
    Since each $\bE_t$ are marginally sub-Gaussian random variables, we invoke Lemma~\ref{lemma:concen} and a union bound to show the previous event occurs with high probability
    \begin{align*}
        \mathbb{P}_0\{\mathcal{E}_{nm}(\alpha)^c\} &\leq m6 \exp\left(3nd - \frac{\tau \alpha^2 d n^2 \epsilon^2_{nm}}{4}\right) \\
        &\leq \exp\left(5nd - \frac{\tau \alpha^2 d n^2  \epsilon_{nm}^2}{4}\right) \\
        &=  \exp\left\{-\left(\frac{\tau \alpha^2 n \epsilon_{nm}^2}{4} - 5\right) d n \right\},
    \end{align*}
    where we used the assumption that $m = O\{\log(n)\}$ and the fact that $d \geq 1$, so that $\log m < nd$ and $\log(6) < nd$ for $n$ large enough.

    Define the set $\mathcal{B}_{nm}(\gamma) = \set{\max_{1 \leq t \leq m} \frobnorm{\bX_t \bX_t^{\top} - \bX_{0t} \bX_{0t}^{\top}} \leq C_{\bx} \sqrt{d} n \gamma}$. Then over the event $\mathcal{E}_{nm}(\alpha)$, the denominator $D_{nm}$ can be lower bounded as follows:
    \begin{align*}
        D_{nm} &\geq \int_{\mathcal{B}_{nm}(\epsilon_{nm})} \prod_{t=1}^m \exp\bigg\{-\frac{\lambda}{2} \frobnorm{\bX_t \bX_t^{\top} - \bX_{0t} \bX_{0t}^{\top}}^2 + \\
        &\qquad\qquad\qquad\qquad\qquad\qquad\lambda \frobket{\bE_t, \bX_t \bX_t^{\top} - \bX_{0t} \bX_{0t}^{\top}} \bigg\} \Pi_{\mathrm{RW}(1)}(\rmd \bX_{1:m} \mid \sigma_0) \\
        &\geq \int_{\mathcal{B}_{nm}(\epsilon_{nm})} \prod_{t=1}^m \exp\bigg\{-\frac{\lambda}{2} \frobnorm{\bX_t \bX_t^{\top} - \bX_{0t} \bX_{0t}^{\top}}^2 - \\
        &\qquad\qquad\lambda \abs*{\frobket*{\bE_t, \frac{\bX_t \bX_t^{\top} - \bX_{0t} \bX_{0t}^{\top}}{\frobnorm{\bX_t \bX_t^{\top} - \bX_{0t}\bX_{0t}^{\top}}}}} \frobnorm{\bX_t \bX_t^{\top} - \bX_{0t}\bX_{0t}^{\top}}  \bigg\} \Pi_{\mathrm{RW}(1)}(\rmd \bX_{1:m} \mid \sigma_0) \\
        &\geq \int_{\mathcal{B}_{nm}(\epsilon_{nm})} \prod_{t=1}^m \exp\bigg\{-\frac{\lambda}{2} \frobnorm{\bX_t \bX_t^{\top} - \bX_{0t} \bX_{0t}^{\top}}^2 -\\
        &\qquad\qquad\qquad\qquad\qquad\qquad\lambda \sqrt{d} \alpha n \epsilon_{nm} \frobnorm{\bX_t \bX_t^{\top} - \bX_{0t}\bX_{0t}^{\top}} \Pi_{\mathrm{RW}(1)}(\rmd \bX_{1:m} \mid \sigma_0) \\
        &\geq \int_{\mathcal{B}_{nm}(\epsilon_{nm})} \prod_{t=1}^m \exp\bigg\{-\frac{\lambda}{2} \frobnorm{\bX_t \bX_t^{\top} - \bX_{0t} \bX_{0t}^{\top}}^2 -\frac{\lambda}{2} \alpha^2 d n^2 \epsilon_{nm}^2 - \\
        &\qquad\qquad\qquad\qquad\qquad\qquad\frac{\lambda}{2} \frobnorm{\bX_t \bX_t^{\top} - \bX_{0t}\bX_{0t}^{\top}}^2 \bigg\}\Pi_{\mathrm{RW}(1)}(\rmd \bX_{1:m} \mid \sigma_0) \\
        &\geq \Pi_{\mathrm{RW}(1)}\{\mathcal{B}_{nm}(\epsilon_{nm}) \mid \sigma_0\} \exp\left\{-\lambda \left(C_{\bx}^2 + \frac{\alpha^2}{2}\right) d n^2 m \epsilon_{nm}^2\right\} \\
        &\geq \exp\left\{-\left(\frac{\lambda(\alpha^2 + 2C_{\bx}^2) + 2C_1}{2} \right) dn^2 m \epsilon^2_{nm}\right\},
    \end{align*}
    where $C_1$ is the constant from Lemma~\ref{lemma:prior_support} and the fourth inequality used the fact that $2ab \leq a^2 + b^2$ for any $a,b > 0$.

    Similarly, over the event $\mathcal{E}_{nm}(\alpha)$, the numerator can be upper bounded as follows:
    \begin{align*}
        N_{nm}&\{\mathcal{B}^c_{nm}(M\epsilon_{nm})\} \\
        &=  \int_{\mathcal{B}_{nm}^c(M \epsilon_{nm})} \prod_{t=1}^m \exp\bigg\{-\frac{\lambda}{2} \frobnorm{\bX_t \bX_t^{\top} - \bX_{0t} \bX_{0t}^{\top}}^2 + \\
        &\qquad\qquad\qquad\qquad\qquad\qquad\lambda \frobket{\bE_t, \bX_t \bX_t^{\top} - \bX_{0t} \bX_{0t}^{\top}} \bigg\} \Pi_{\mathrm{RW}(1)}(\rmd \bX_{1:m} \mid \sigma_0)  \\
        &\leq \int_{\mathcal{B}_{nm}^c(M \epsilon_{nm})} \prod_{t=1}^m \exp\bigg\{-\frac{\lambda}{2} \frobnorm{\bX_t \bX_t^{\top} - \bX_{0t} \bX_{0t}^{\top}}^2 + \\
        &\qquad\qquad\qquad\qquad\qquad\qquad \lambda \sqrt{d} \alpha n \epsilon_{nm} \frobnorm{\bX_t \bX_t^{\top} - \bX_{0t} \bX_{0t}^{\top}} \bigg\} \Pi_{\mathrm{RW}(1)}(\rmd \bX_{1:m} \mid \sigma_0) \\
        &\leq \int_{\mathcal{B}_{nm}^c(M \epsilon_{nm})} \prod_{t=1}^m \exp\bigg\{-\frac{\lambda}{2} \frobnorm{\bX_t \bX_t^{\top} - \bX_{0t} \bX_{0t}^{\top}}^2 + \lambda 2 \alpha^2 d n^2 \epsilon_{nm}^2 + \\
        &\qquad\qquad\qquad\qquad\frac{\lambda}{8} \frobnorm{\bX_t \bX_t^{\top} - \bX_{0t} \bX_{0t}^{\top}}^2 \bigg\} \Pi_{\mathrm{RW}(1)}(\rmd \bX_{1:m} \mid \sigma_0) \\
        &\leq \exp\left\{-\lambda \left(\frac{3M^2C_{\bx}^2}{8} - 2 \alpha^2\right) d n^2 m \epsilon_{nm}^2 \right\},
    \end{align*}
    where the third inequality used the fact that $ab \leq 2a^2 + b^2/8$ for any $a,b > 0$. 

    We proceed to bound $\mathbb{E}_0[\Pi_{\lambda}\{\mathcal{B}_{nm}(M\epsilon_{nm}) \mid \bY_{1:m}\}]$ as follows:
    \begin{align*}
        \mathbb{E}_0&[\Pi_{\lambda}\{\mathcal{B}_{nm}^c(M\epsilon_{nm}) \mid \bY_{1:m}\}] \\
        &\leq\mathbb{E}_0\left[\mathbbm{1}\{\bY_{1:m} \in \mathcal{E}_{nm}(\alpha)\}\ \frac{N_{nm}\{\mathcal{B}_{nm}^c(M\epsilon_{nm})\}}{D_{nm}}\right] + \mathbb{P}_0\{\mathcal{E}_{nm}^c(\alpha)\} \\
        &\leq\exp\left\{\left(\frac{\lambda(\alpha^2 + 2C_{\bx}^2) + 2C_1}{2} \right) d n^2 m \epsilon^2_{nm} - \lambda \left(\frac{3}{8}M^2C_{\bx}^2 - 2 \alpha^2\right) d n^2 m \epsilon_{nm}^2 \right\} \\
        &\qquad + \exp\left\{-\left(\frac{\tau \alpha^2 n \epsilon^2_{nm}}{4} - 5 \right)dn \right\} \\
        &\leq\exp\left\{-\lambda \left(\frac{3}{8} M^2 C_{\bx}^2 - \frac{5\alpha^2 + 2C_{\bx}^2 + 2C_1/\lambda}{2} \right) d n^2 m \epsilon^2_{nm}\right\} +  \exp\left\{-\left(\frac{\tau \alpha^2 n \epsilon^2_{nm}}{4} - 5 \right)dn \right\}.
    \end{align*}
    Recall the assumption that $m = O(\log n)$ so that  $n\epsilon_{nm}^2 \gtrsim 1$, which ensures the second term goes to zero for $\alpha$ sufficiently large. Hence, taking $\alpha = \sqrt{\{M^2 C_{\bx}^2 - 2(2C_{\bx}^2+2C_1/\lambda)\}/10}$, for sufficiently large $M$, we see that
    \begin{align*}
        \mathbb{E}_0&[\Pi_{\lambda}\{\mathcal{B}_{nm}^c(M\epsilon_{nm}) \mid \bY_{1:m}\}]  \leq \exp\left(-\frac{\lambda}{16} C_{\bx}^2 M^2 d n^2 m \epsilon_{nm}^2\right) + \exp\left(-\frac{\tau}{80} C_{\bx}^2 M^2  dn^2  \epsilon_{nm}^2\right) \\
        &\leq 2 \exp\left\{-\min\left(\frac{\lambda}{16}, \frac{\tau}{80}\right) C_{\bx}^2 M^2 d n^2 \epsilon_{nm}^2\right\}.
    \end{align*}
    
    Next, define the set $\mathcal{U}_{nm}(\gamma) = \{ \max_{1 \leq t \leq m} \inf_{\bW_t \in \mathcal{O}_d} \frobnorm{\bX_t - \bX_{0t} \bW_t}^2 \leq C_{\bx}^2 nd \gamma^2\}$. Using the assumption that $\bX_{0t}$ is full rank with minimum singular value $\tilde{\sigma}_{\text{min}}(\bX_{0t}) \geq \tilde{\sigma} \sqrt{n}$ for all $1 \leq t \leq n$ and Lemma~\ref{lemma:tu2016}, we see that 
    \[
        \max_{1 \leq t \leq m} \inf_{\bW_t \in \mathcal{O}_d} \frobnorm{\bX_t - \bX_{0t} \bW_t}^2 \leq  \frac{\tilde{\sigma}}{n} \max_{1 \leq t \leq m} \frobnorm{\bX_t \bX_t^{\top} - \bX_{0t} \bX_{0t}^{\top}}^2. 
    \]
    As such,
    \[
        \frac{1}{n} \max_{1 \leq t \leq m} \inf_{\bW_t \in \mathcal{O}_d} \frobnorm{\bX_t - \bX_{0t} \bW_t}^2 \leq  \frac{\tilde{\sigma}}{n^2} \max_{1 \leq t \leq m} \frobnorm{\bX_t \bX_t^{\top} - \bX_{0t} \bX_{0t}^{\top}}^2,
    \]
    which implies $\mathcal{U}_{nm}^c(M\epsilon_{nm}) \subset \mathcal{B}_{nm}^c\{(M/\tilde{\sigma})^{1/2} \epsilon_{nm}\}$. Thus,
    \begin{align*}
        \mathbb{E}_0[\Pi_{\lambda}\{\mathcal{U}_{nm}^c(M\epsilon_{nm}) \mid \bY_{1:m}\}] &\leq  \mathbb{E}_0[\Pi_{\lambda}\{\mathcal{B}_{nm}^c\{(M/\tilde{\sigma})^{1/2}\epsilon_{nm}\} \mid \bY_{1:m}\}] \\
        &\leq 2 \exp\left\{-\frac{1}{\tilde{\sigma}}\min\left(\frac{\lambda}{16}, \frac{\tau}{80}\right)C_{\bx}^2 M d n^2 \epsilon_{nm}^2\right\},
    \end{align*}
    which completes the proof.
\end{proof}

\section{Additional Simulation Study Results}

\subsection{Visualizations of the Estimated and True Latent Trajectories}\label{sec:vis_sim}

In this section, we provide an explicit visualization comparing the estimated and true latent trajectories from a single replication of the simulation study described in Section~\ref{sec:simulation} of the main text. Specifically, for $\nu = 0.5$ and $\nu = 2.5$, we generated a dynamic network with $n = 200$ nodes, $m = 50$ time points, and an expected density of $0.2$ from the same data-generating process described in the simulation study of the main text. Then, using each simulated network, we estimated  GB-DASE with an RW(1) and RW(2) prior and the competing ASE-based methods using the same procedures and hyperparameter settings described in the simulation study in Section~\ref{sec:simulation} of the main text. 

Figure~\ref{fig:trajectories_exp} displays the first dimension of the inferred and true latent trajectories of three nodes when $\nu = 0.5$ and $\nu = 2.5$, respectively. These visualizations highlight some of the causes behind the performance differences between the competing methods in the simulation study in Section~\ref{sec:simulation} of the main text. OMNI tends to over-smooth the latent trajectories, causing the estimates to be biased above or below the peaks and troughs of the true latent trajectories. As previously observed, this behavior is caused by OMNI sharing information globally when the networks are only similar locally. The ASE estimator suffers from the opposite behavior, as it shares no information across networks. As such, the ASE estimator produces sporadic latent trajectory estimates that contain large jumps that overfit to the individual networks. The GB-DASE estimators behave like locally-smoothed versions of the ASE estimator that better match the true latent trajectories. In addition, the GB-DASE estimates are smoother under the RW(2) prior than the RW(1) prior, which is beneficial when $\nu = 2.5$. Lastly, the 95\% generalized posterior credible intervals from both GB-DASE methods often cover the true values while maintaining reasonable widths.

\begin{figure}[htb]
\centering
\begin{subfigure}[b]{\textwidth}
    \centering 
    \includegraphics[width=\textwidth, keepaspectratio]{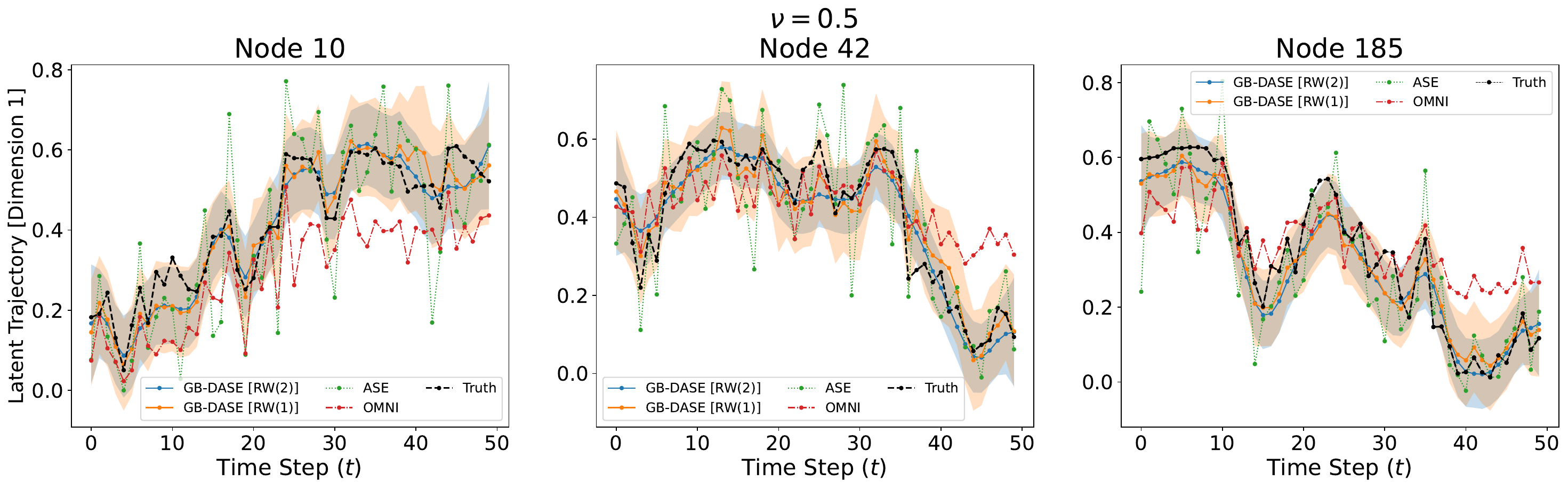}
\end{subfigure}
\begin{subfigure}[b]{\textwidth}
    \centering 
    \includegraphics[width=\textwidth, keepaspectratio]{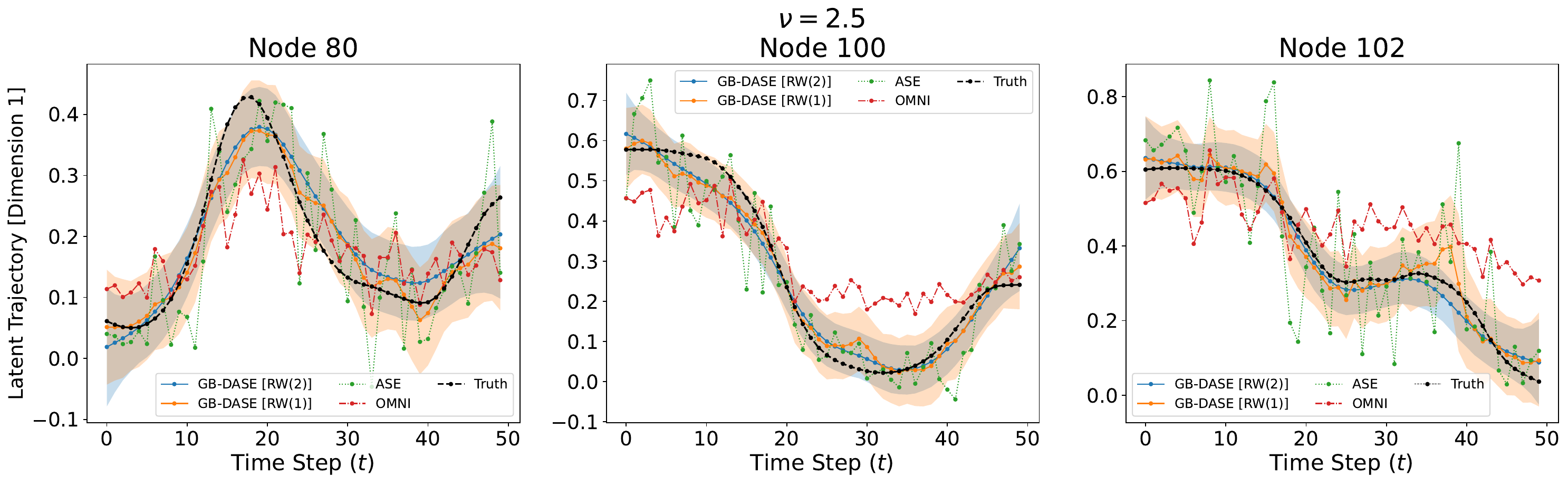}
\end{subfigure}
    \caption{The first dimension of the true and inferred latent trajectories for three nodes from a network with $n = 200$, $m = 50$, an expected density of $0.2$ under two smoothness settings: (Top) $\nu = 0.5$ and (Bottom) $\nu = 2.5$. The colored curves are the point estimates from the four competitors. The shaded regions indicate 95\% pointwise credible intervals from GB-DASE with an RW(2) prior (blue) and an RW(1) prior (orange), respectively.}
    \label{fig:trajectories_exp}
\end{figure}

\subsection{Forecasting Performance}\label{sec:forecast_sim}

This section contains a simulation study comparing the $k$-step ahead forecasting performance of GB-DASE with four ASE-based competitors when the true forecasting function is a linear function.  In particular, the proposed GB-DASE methods use the $k$-step ahead RW($r$)-generalized posterior predictive distribution's expectation defined in Equation~(\ref{eq:gen_forecast}) of the main text as a $k$-step ahead forecast. We form a simple Monte Carlo estimate of this expectation using 2,500 samples from the GB-DASE posterior obtained with Algorithm~\ref{algo:gibbs} after a burn-in of 2,500 iterations. We considered the same four competitors as in the real-data application in Section~\ref{sec:real_data} of the main text, that is, ASE, OMNI, the unfolded ASE (UASE), and the multiple ASE (MASE). Since these existing ASE-based methods do not provide $k$-step ahead forecasts, we use the edge probability estimates at the last observed time point as a forecast for the edge probability at future time points, that is, we use $\hat{\mathbb{E}}(\bY_m)$ as a forecast for $\bY_{m+k}$. For all estimators, we fixed the latent space dimension to the true dimension.

%The first two methods use the ASE and OMNI  estimators defined in Section~\ref{sec:simulation} of the main text, that is, $\hat{\mathbb{E}}(\bY_{m}) = \hat{\bX}_{m}^{(\text{ASE})} \hat{\bX}_{m}^{(\text{ASE}) \, \top}$ and $\hat{\mathbb{E}}(\bY_{m}) = \hat{\bX}_{m}^{(\text{OMNI})} \hat{\bX}_{m}^{(\text{OMNI}) \, \top}$. The other two competitors use the unfolded ASE (UASE)~\citepSup{jones2021, gallagher2021}  and the multiple ASE (MASE)~\citeSup{arroyo2021}, which impose shared structure on the expected adjacency matrices. In particular, UASE assumes the networks come from a multilayer random dot product graph, that is, $\mathbb{E}_0(\bY_{t}) = \bV \mathbf{R}^{(t)} \bU^{(t) \, \top}$ for $\bU \in \Reals{n \times d}, \mathbf{R}^{(t)} \in \Reals{d \times d}$, and $\bV^{(t)} \in \Reals{n \times d}$. Similarly, MASE assumes the networks comes from the COSIE model, that is, $\mathbb{E}_0(\bY_t) = \bV \mathbf{R}^{(t)} \bV^{\top}$ for $\bV \in \Reals{n \times d}$ and $\mathbf{R}^{(t)}\in \Reals{d \times d}$ so that the expectations share a common invariant subspace.  

We generate synthetic networks from a Bernoulli dynamic RDPG with independent edges and $d = 2$ latent space dimensions defined in Equation~(\ref{eq:bern_rdpg}) of the main text for $m + \ell$ time points. We use the first $m$ networks to estimate the models and the last $\ell$ networks to evaluate the $k$-step ahead forecasting performance. In all experiments, we set $\ell = 5$. For the first $m - q$ time points, we set the true latent trajectories to the values at times $t = 1, \dots, m-q$ of a function parameterized in terms of a cubic B-spline basis with three knots equally spaced on the interval $[0, m - q]$, so that
\[
    \bx_{it} = \rho \, d^{-1/2} \, (\bw^{\top}_{i1} \mathbf{b}(t), \dots, \bw^{\top}_{id} \mathbf{b}(t))^{\top} \qquad (1 \leq i \leq n, 1 \leq t \leq m - q),
\]
where the scalar $\rho \in [0,1]$, $\mathbf{b}(\cdot)$ is a five-dimensional collection of B-spline basis functions, and each $\bw_{ih} \in \Reals{5}$ consists of basis weights. In the subsequent experiments, we set $\rho$ to fix the overall density of the networks. To ensure each $\bx_{it} \in \mathcal{X} = \{\bx \in \Reals{d}\, : \, \norm{\bx}_2 \leq 1, \bx \succeq 0\}$, we sample $\bw_{ih} \iidsim \text{Dirichlet}(1/5, \dots, 1/5)$ for $1 \leq i \leq n$ and $1 \leq h \leq d$. To set each latent trajectory's values at the remaining $q + \ell$ time points, we linearly extrapolated the latent trajectory forward in time so that the true forecasting function is linear. In all experiments, we set $q = 10$. Figure~\ref{fig:sim_forecast_viz} displays edge probabilities that this procedure produces. To compare the $k$-step ahead forecasts, we evaluate the error for recovery of the $k$-step ahead edge probabilities:
\begin{align*}
    \text{RMSE}_k &= \left[\frac{2}{n(n-1)}\sum_{i < j} \left\{\bx_{i(m+k)}^{\top} \bx_{j(m+k)} - \hat{\mathbb{P}}(Y_{ij,m+k} = 1 \mid \bY_{1:m})\right\}^2\right]^{1/2}, 
\end{align*}
where $\hat{\mathbb{P}}(Y_{ij,m+k} = 1 \mid \bY_{1:m})$ is the forecasted probability of an edge between $i$ and $j$ at time $m+k$, that is, $\hat{\mathbb{E}}_{Q_{\lambda}}^{(r)}(Y_{ij, m+k} \mid \bY_{1:m})$ for GB-DASE and $\hat{\mathbb{E}}(Y_{ij,m})$ for the ASE-based methods.

Figure~\ref{fig:sim_forecast} displays the results for networks with $n = 200$ nodes. For each experiment, we varied the number of time points $m \in \set{100, 150}$ for networks with expected densities of 0.1, 0.2, and 0.3. In all settings, the GB-DASE methods perform best. As expected, GB-DASE with an RW(2) prior outperforms GB-DASE with an RW(1) prior when the signal-to-noise ratio is high enough because it estimates a linear forecast function that matches the true linear forecasts. %Specifically, at the lowest expected density of 0.1 the forecasting errors of both GB-DASE methods coincide. 
All methods improved slightly as $m$ increased, with their relative performance unchanged. Lastly, for networks with expected densities of 0.2 and 0.3, we see that the error of GB-DASE with an RW(2) prior is roughly constant as the forecasting window $k$ increases, unlike the competing methods. This behavior is because this method estimates the correct linear forecasting rule. To demonstrate, Figure~\ref{fig:sim_forecast_viz} displays the inferred and true edge probabilities between three dyads for a network with $n = 200$, $m = 100$, and an expected density of $0.2$. We see that the linear forecasts from GB-DASE with an RW(2) prior do well in tracking the true edge probabilities, unlike the constant forecasts from the competing methods. However, the 95\% credible bands from GB-DASE with an RW(1) prior still cover the truth. %We do not always expect linear forecasts to perform best in practice. As such, we compare the methods' forecasting performance using a real-world conflict network in the next section.

\begin{figure}[htbp]
\centering \includegraphics[width=\textwidth, keepaspectratio]{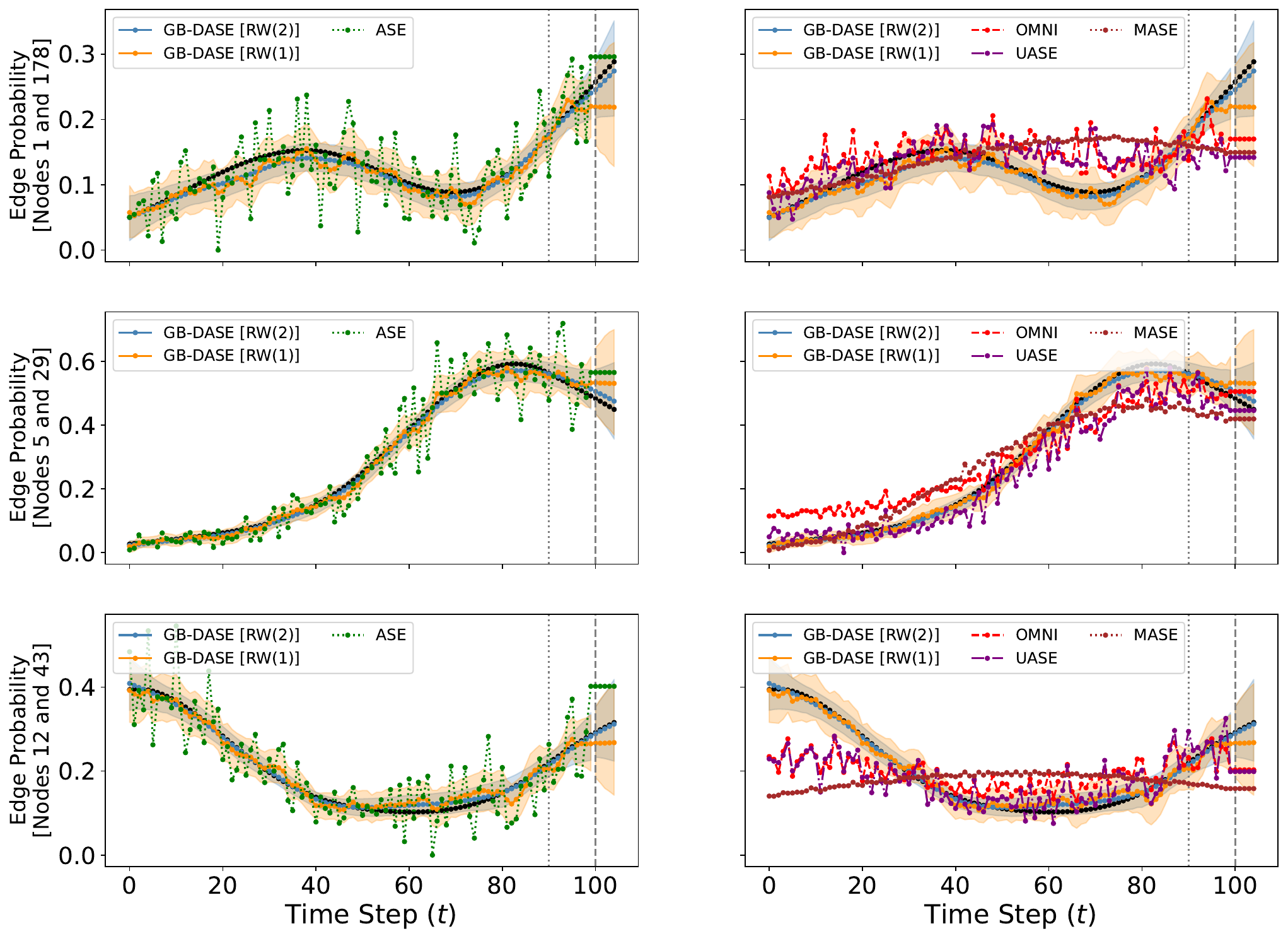}
    \caption{The inferred and true edge probabilities between three dyads from a network with $n = 200$, $m = 100$, and an expected density of $0.2$. The colored curves are the point estimates from the six competitors. The shaded regions indicate 95\% pointwise credible intervals from  GB-DASE with an RW(2) prior (blue) and GB-DASE with an RW(1) prior (orange), respectively. The vertical dotted line indicates the start of linear extrapolation ($t = m-q+1$). The vertical dashed line indicates the time at which forecasting begins ($t = m+1$).}
    \label{fig:sim_forecast_viz}
\end{figure}

\begin{figure}[htbp]
\centering \includegraphics[width=\textwidth, keepaspectratio]{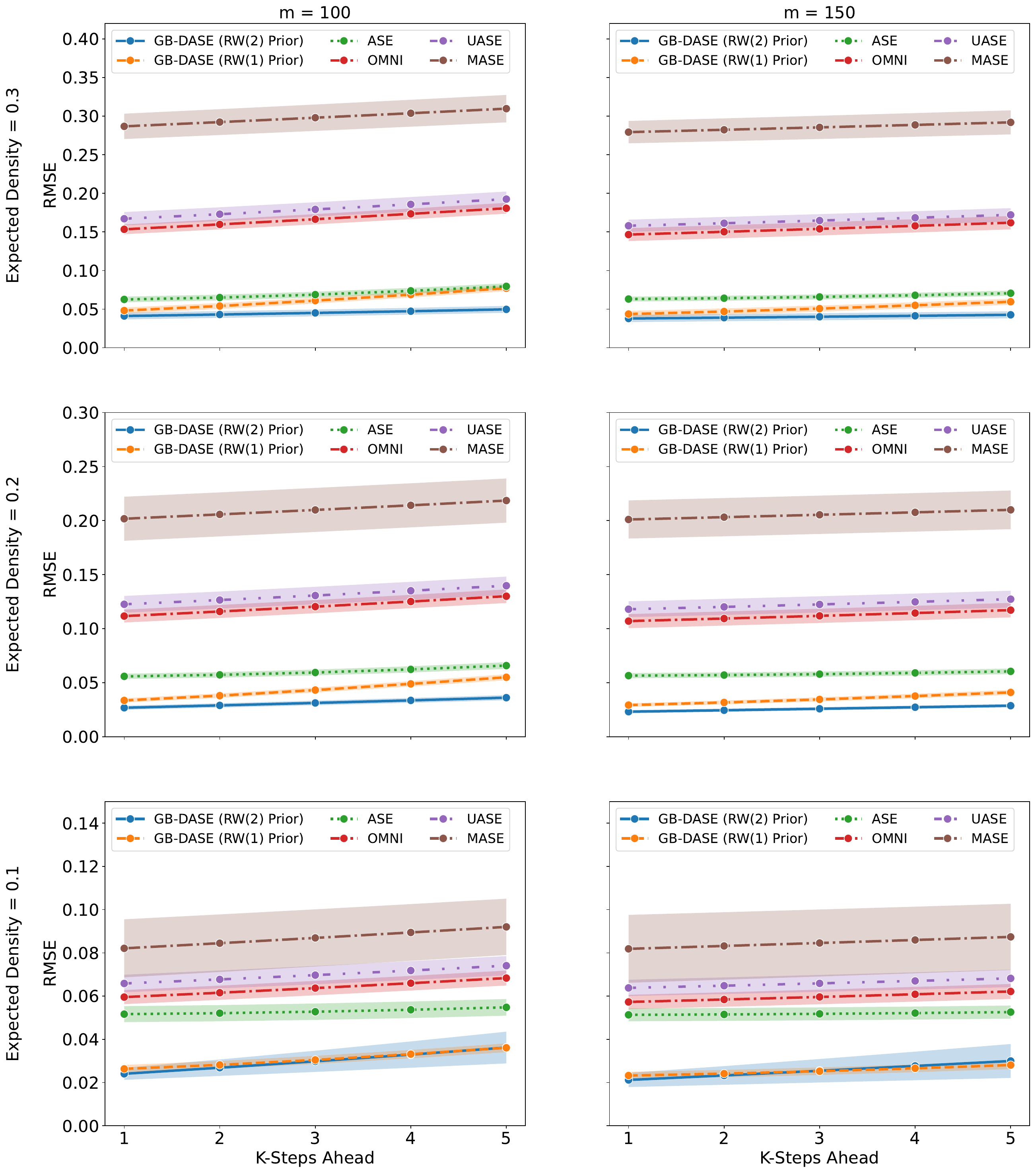}
    \caption{RMSEs of the $k$-step ahead forecasts for networks with a fixed $n = 200$ and varying $m$ and expected densities. The curves and shaded regions indicate averages and one standard deviations over 50 independent replicates, respectively.}
    \label{fig:sim_forecast}
\end{figure}

\section{Additional Results from the Real Data Application}

\subsection{Goodness-of-Fit to the Degree Distribution}\label{sec:gof}

For a more nuanced in-sample goodness-of-fit comparison on the conflict network analyzed in Section~\ref{sec:real_data} of the main text, we quantify each competing method's ability to reconstruct the observed dynamic network's degree distribution using the graphical goodness-of-fit test proposed by \citetSup{hunter2008}. Because this graphical test requires the specification of a complete data-generating process, for this comparison, we will assume that the observed networks come from various independent edge inhomogeneous Bernoulli random graph models with edge probabilities assumed by the corresponding methods. As such, the primary differences highlighted by the graphical test will be due to the edge probability matrices estimated by each method. Specifically, for each ASE-based method, we simulate 100 networks from their associated data-generating model with parameters equal to their point estimates. Based on these simulated networks, we report the empirical degree distribution, that is, the distribution of the number of nodes with a specific degree observed over all time points. For the GB-DASE methods, we calculate the generalized posterior predictive degree distribution using a working data-generating model from the dynamic Bernoulli RDPG with independent edges in Equation~(\ref{eq:bern_rdpg}) of the main text.

Figure~\ref{fig:deg_predictive} depicts the results. We see that ASE and OMNI do a poor job of reconstructing the degree distribution, with both methods producing too many low-degree nodes. In contrast, UASE, MASE, and both GB-DASE methods capture the observed degree distribution well. The generalized posterior predictive degree distributions have wider simulation bands than the empirical distributions. This difference is because the generalized posterior predictive distribution expresses both aleatoric uncertainty from the data-generating model and epistemic uncertainty in the parameters from the generalized posterior. In contrast, the  empirical distributions of the ASE-based methods only reflect aleatoric uncertainty. This comparison highlights a benefit of GB-DASE, which is providing rigorous parameter uncertainty.

\begin{figure}[htbp]
\centering \includegraphics[width=\textwidth, keepaspectratio]{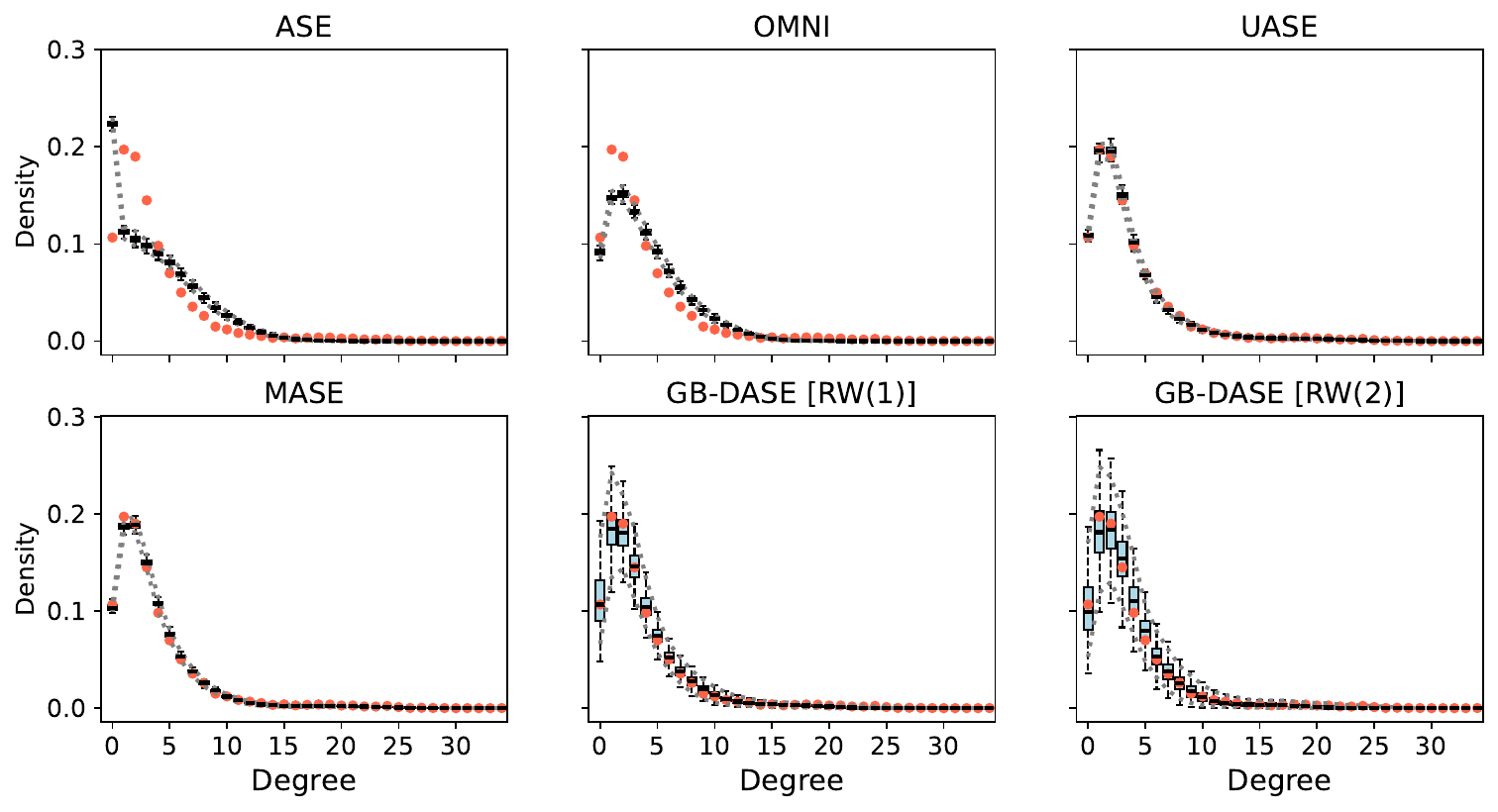}
    \caption{Graphical goodness-of-fit plots for the competing methods on the conflict networks. The boxplots indicate the empirical  (ASE, OMNI, UASE, MASE) or  generalized posterior predictive (GB-DASE) degree distribution, e.g., the number of degree $k$ nodes observed over all time points. The red dots represent the observed statistics, and the region between the dashed lines indicates the range in which 95\% of the simulated statistics fall. For the GB-DASE methods, these regions indicate 95\% pointwise credible intervals.}
    \label{fig:deg_predictive}
\end{figure}

\subsection{Additional Figures from the Real Data Application}\label{sec:add_figures}

Figure~\ref{fig:ase_dhat} displays embedding dimensions selected by the profile-likelihood method of \citetSup{zhu2006} when applied separately to the 157 conflict networks. Based on this plot, we chose an embedding dimension of four for GB-DASE because the majority of estimates were less than or equal to four. Figures~\ref{fig:un_is_rw1} and \ref{fig:pl_us_rw1} show the edge probabilities estimated by GB-DASE with an RW(1) prior between Israel and the United Nations, and between the Palestinian Territories and the United States, respectively. As discussed in the main text, the estimates are less smooth than those from GB-DASE with an RW(2) prior. In particular, they tend to better capture sudden changes in the in-sample time series, which explains the improved in-sample performance of GB-DASE with an RW(1) prior. In contrast, GB-DASE with an RW(1) prior produces a constant forecast, which in this scenario underperforms the linear forecast produced by GB-DASE with an RW(2) prior. Moreover, the widths of the 95\% credible intervals for these forecasts are substantially wider than those from GB-DASE with an RW(2) prior.

\begin{figure}[tb]
\centering \includegraphics[width=0.75\textwidth, keepaspectratio]{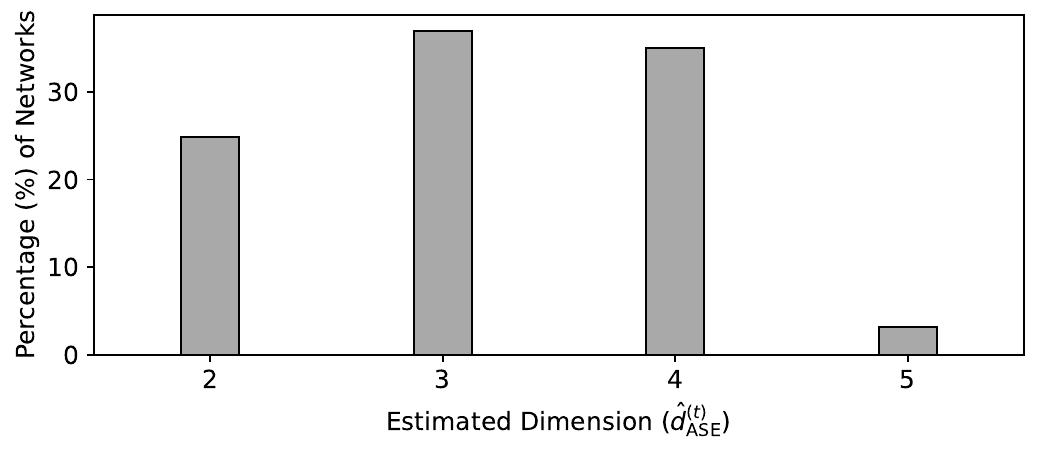}
    \caption{Percentage that the profile-likelihood method selected a given embedding dimension when applied separately to the 157 conflict networks.}
    \label{fig:ase_dhat}
\end{figure}

\begin{figure}[htbp]
\centering \includegraphics[height=0.34\textheight, width=\textwidth, keepaspectratio]{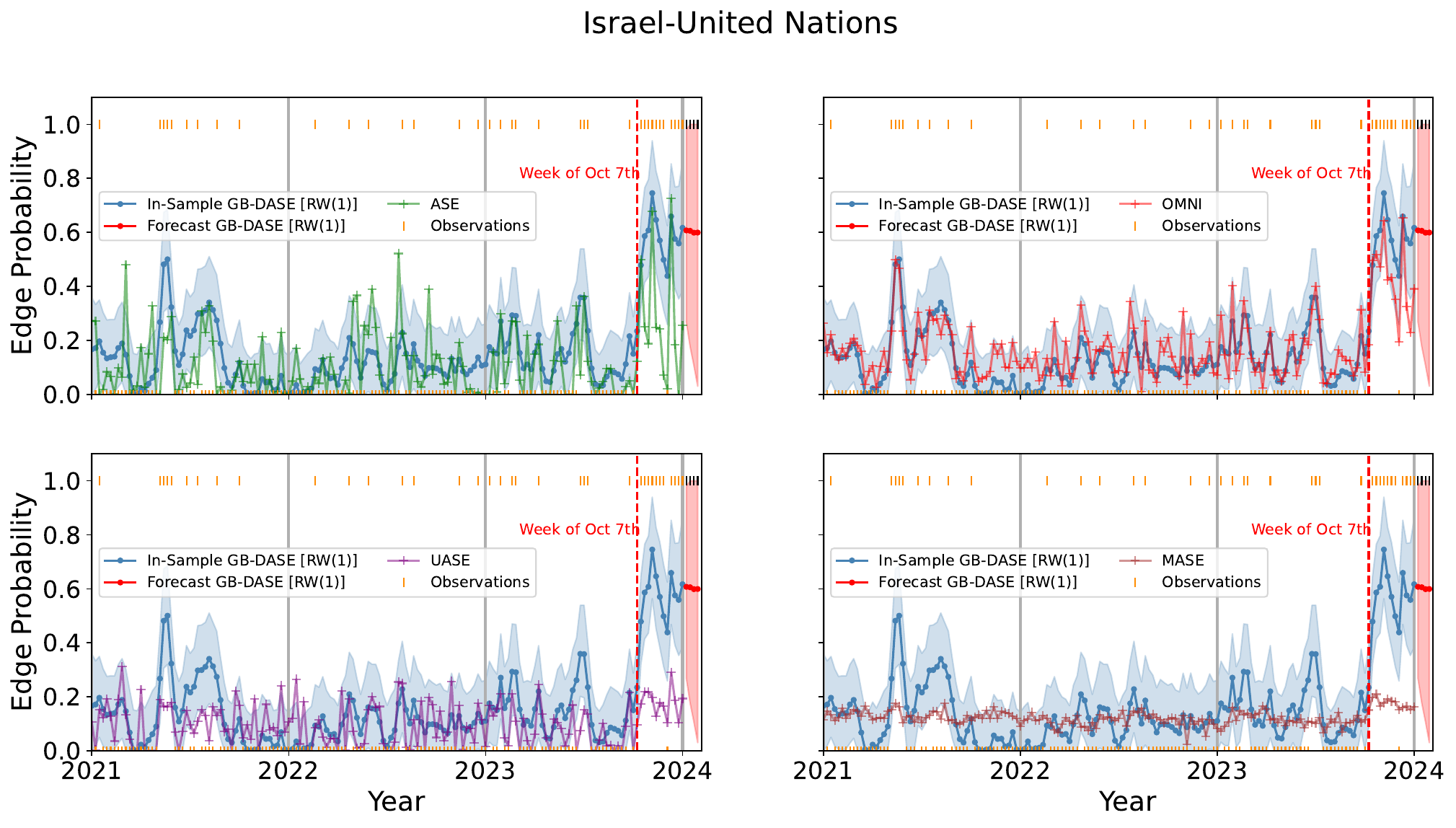}
    \caption{Estimated conflict probabilities between Israel and the United Nations. The curves denote point estimates. The blue and red shaded regions indicate 95\% pointwise credible intervals for in-sample and out-of-sample probabilities, respectively. The orange and black ticks indicate whether a conflict occurred ($y_{ij,t} = 1$) or did not occur ($y_{ij,t} = 0$) on a given week. Orange and black ticks indicate in-sample and out-of-sample weeks, respectively. The dashed horizontal line denotes the week of October 7th, 2023.}
    \label{fig:un_is_rw1}
\end{figure}

\begin{figure}[htbp]
\centering \includegraphics[height=0.34\textheight, width=\textwidth, keepaspectratio]{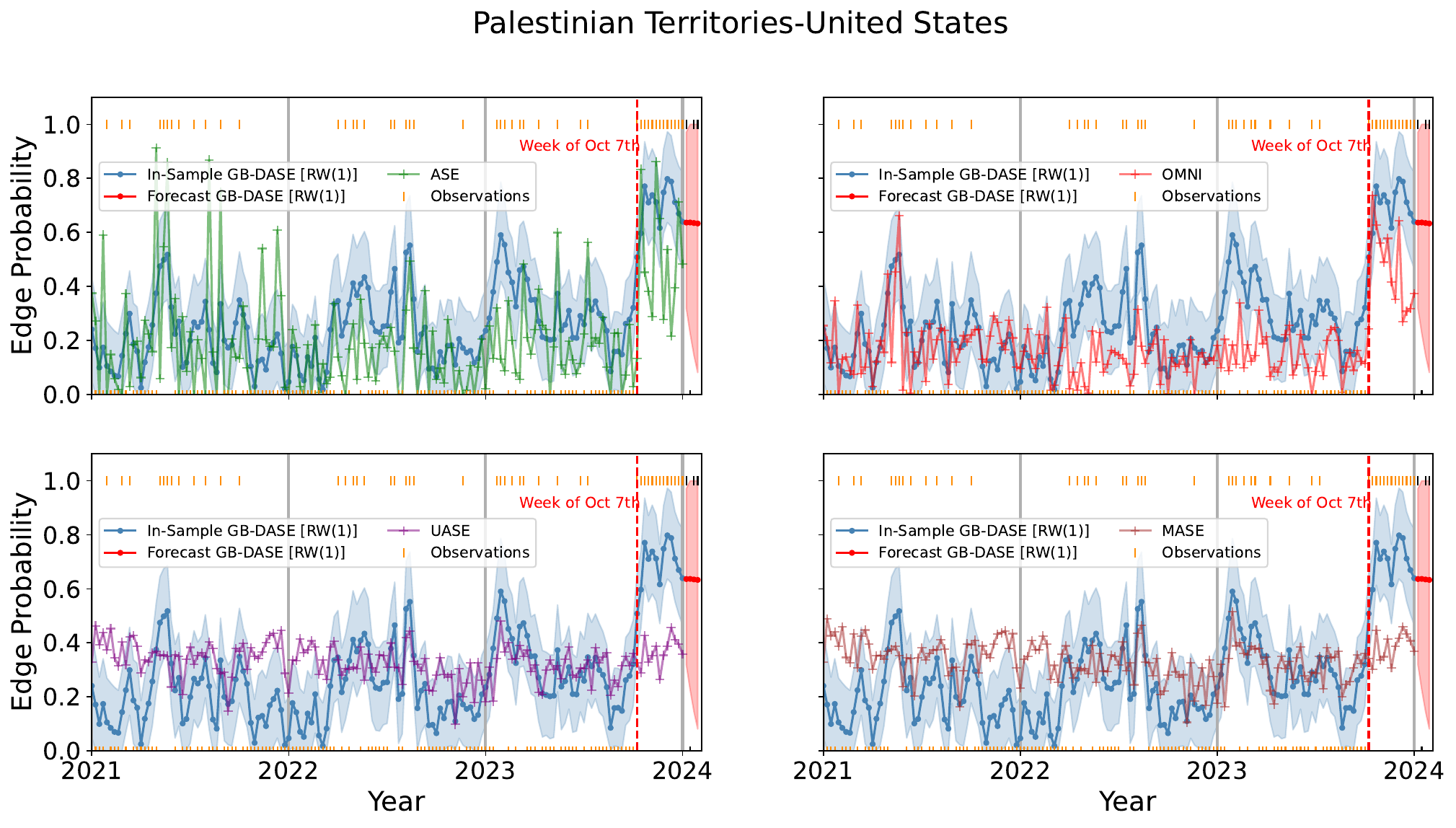}
    \caption{Estimated conflict probabilities between the Palestinian territories and the United States. The curves denote point estimates. The blue and red shaded regions indicate 95\% pointwise credible intervals for in-sample and out-of-sample probabilities, respectively. The orange and black ticks indicate whether a conflict occurred ($y_{ij,t} = 1$) or did not occur ($y_{ij,t} = 0$) on a given week. Orange and black ticks indicate in-sample and out-of-sample weeks, respectively. The dashed horizontal line denotes the week of October 7th, 2023.}
    \label{fig:pl_us_rw1}
\end{figure}

\clearpage

\bibliographystyleSup{agsm}
\bibliographySup{bibliography}

\end{document}